\documentclass[11pt,A4paper]{article}

\usepackage{amsmath}
\usepackage{amssymb}
\usepackage{amsfonts}
\usepackage{amsthm}
\usepackage{mathtools}
\mathtoolsset{%
}
\usepackage{mathscinet}

\usepackage[utf8]{inputenc}
\usepackage[T1]{fontenc}

\usepackage{libertine}
\usepackage[libertine,libaltvw,cmintegrals,varbb]{newtxmath}

\usepackage{dsfont}

\usepackage{mathrsfs}

\usepackage[%
cal=cm,
]
{mathalfa}

\usepackage[dvipsnames,svgnames]{xcolor}
\colorlet{MyBlue}{DodgerBlue!60!Black}
\colorlet{MyGreen}{DarkGreen!85!Black}

\usepackage{fullpage}

\usepackage{setspace}
\usepackage[font=small,labelfont=bf]{caption}
\usepackage{subfigure}
\usepackage{tikz}
\usetikzlibrary{calc,patterns}

\usepackage{acronym}
\usepackage{booktabs}       
\usepackage{latexsym}
\usepackage{paralist}
\usepackage{wasysym}
\usepackage{xspace}

\usepackage[authoryear,compress,comma]{natbib}

\usepackage{hyperref}
\hypersetup{
colorlinks=true,
linktocpage=true,
pdfstartview=FitH,
breaklinks=true,
pdfpagemode=UseNone,
pageanchor=true,
pdfpagemode=UseOutlines,
plainpages=false,
bookmarksnumbered,
bookmarksopen=false,
bookmarksopenlevel=1,
hypertexnames=true,
pdfhighlight=/O,
urlcolor=MyBlue!60!black,linkcolor=MyBlue!70!black,citecolor=DarkGreen!70!black, 
pdftitle={},
pdfauthor={},
pdfsubject={},
pdfkeywords={},
pdfcreator={pdfLaTeX},
pdfproducer={LaTeX with hyperref}
}

\numberwithin{equation}{section}  
\usepackage[sort&compress,capitalize,nameinlink]{cleveref}
\crefname{app}{Appendix}{Appendices}

\crefrangeformat{equation}{\upshape(#3#1#4)\textendash(#5#2#6)}

\newcommand{\R}{\mathbb{R}}

\DeclareMathOperator{\card}{card}

\DeclareMathOperator{\degr}{degr}

\DeclareMathOperator{\ex}{\debug{\mathbb{E}}}

\DeclareMathOperator{\hull}{\debug \Delta}

\DeclareMathOperator{\proj}{proj}
\DeclareMathOperator{\prob}{\debug{\mathbb{P}}}

\DeclareMathOperator{\val}{\debug{val}}
\DeclareMathOperator{\Val}{\debug{Val}}

\DeclarePairedDelimiter{\braces}{\{}{\}}
\DeclarePairedDelimiter{\bracks}{[}{]}
\DeclarePairedDelimiter{\parens}{(}{)}
\DeclarePairedDelimiter{\angles}{\langle}{\rangle}

\DeclarePairedDelimiterX{\braket}[2]{\langle}{\rangle}{#1,#2}

\DeclarePairedDelimiterX{\inner}[2]{\langle}{\rangle}{#1,#2}
\DeclarePairedDelimiterX{\setdef}[2]{\{}{\}}{#1:#2}

\DeclarePairedDelimiterXPP{\probof}[1]{\prob}{(}{)}{}{%

#1}

\DeclarePairedDelimiterXPP{\exof}[1]{\ex}{[}{]}{}{%

#1}

\newcommand{\textpar}[1]{\textup(#1\textup)}

\usepackage[textwidth=30mm]{todonotes}

\newcommand{\debug}[1]{#1}

\theoremstyle{plain}
\newtheorem{theorem}{Theorem}
\newtheorem{corollary}[theorem]{Corollary}
\newtheorem*{corollary*}{Corollary}
\newtheorem{lemma}[theorem]{Lemma}
\newtheorem{proposition}[theorem]{Proposition}

\theoremstyle{definition}
\newtheorem{definition}[theorem]{Definition}
\newtheorem*{definition*}{Definition}

\newtheorem*{hypothesis*}{Hypothesis}

\theoremstyle{remark}
\newtheorem{remark}{Remark}
\newtheorem*{remark*}{Remark}
\newtheorem*{notation*}{Notational remark}
\newtheorem{example}{Example}

\numberwithin{theorem}{section}
\numberwithin{remark}{section}
\numberwithin{example}{section}

\newcommand{\game}{\debug \Gamma}
\newcommand{\history}{\debug h}
\newcommand{\histories}{\mathscr{\debug H}}

\newcommand{\mixed}{\debug \sigma}

\newcommand{\pure}{\debug s}

\newcommand{\nPures}{\debug S}
\newcommand{\pures}{\mathscr{\debug \nPures}}

\newcommand{\pay}{\debug g}

\newcommand{\cycle}{\debug \gamma}
\newcommand{\cyclealt}{\cycle'}
\newcommand{\edge}{\debug e}

\newcommand{\edges}{\mathscr{\debug E}}
\newcommand{\graph}{\mathscr{\debug G}}

\newcommand{\sink}{\debug D}

\newcommand{\source}{\debug O}

\newcommand{\subtree}{\debug \tree}
\newcommand{\tree}{\mathscr{\debug T}}
\newcommand{\trees}{\debug{\mathbb{T}}}

\newcommand{\vertex}{\debug v}
\newcommand{\vertexalt}{\debug u}

\newcommand{\vertices}{\mathscr{\debug V}}

\newcommand{\simplex}{\debug \Delta}

\newcommand{\altdiff}{\debug \Phi}
\newcommand{\bdfs}{\debug \alpha}
\newcommand{\contnet}{\debug Q}
\newcommand{\cycletime}{\debug \tau}
\newcommand{\diffedge}{\debug \Lambda}
\newcommand{\ebd}{\debug{\mixedh^{\ast}}}
\newcommand{\ebdcond}[1]{\debug{\mixedh_{#1}^{\ast}}}
\newcommand{\edgel}{\debug j}
\newcommand{\euler}{\mathscr{\debug P}}
\newcommand{\eulerstr}{\debug{\mixed^{\ast}}}
\newcommand{\eulertime}{\debug \theta}
\newcommand{\graphs}{\debug{\mathbb{G}}}
\newcommand{\halfpaths}{\debug m}
\newcommand{\length}{\debug \lambda}
\newcommand{\lengthprof}{\debug{\boldsymbol \length}}
\newcommand{\maxdegree}{\debug \delta}
\newcommand{\mixedh}{\debug \varepsilon}
\newcommand{\neighbor}{\debug{\mathscr{N}}}
\newcommand{\nout}{\debug n}
\newcommand{\npaths}{\debug n}
\newcommand{\ppath}{\debug i}
\newcommand{\ppathalt}{\debug k}
\newcommand{\ppathsec}{\debug{\ppathalt'}}
\newcommand{\per}{\debug t}
\newcommand{\proba}{\debug p}
\newcommand{\probas}{\boldsymbol{\proba}}
\newcommand{\probbin}{\debug \zeta}

\newcommand{\subedges}{\debug I}
\newcommand{\totlength}{\debug \mu}
\newcommand{\ud}{\debug{\mixedh^{\textup{U}}}}
\newcommand{\vertexedge}{\debug z}
\newcommand{\totedges}{\debug L}

\newcommand{\TG}[1]{\todo[color=Goldenrod!50,author=\textbf{Tristan},inline]{\small #1\\}}

\newacro{BDFS}{biased depth-first strategy}
\newacro{DFS}{depth-first strategy}
\newacro{DSG}{deterministic search game}
\newacro{EBD}{equal branching distribution}
\newacro{ES}{Eulerian strategy}
\newacro{HSG}{hide-search game}
\newacro{NE}{Nash equilibrium}
\newacro{SG}{search game}
\newacro{SSG}{stochastic search game}
\newacro{UCPS}{uniform Chinese postman strategy}
\newacro{UD}{uniform distribution}
\newacro{UDFS}{uniform depth-first strategy}
\newacro{UES}{uniform Eulerian strategy}

\acrodefplural{BDFS}[BDFSs]{biased depth-first strategies}
\acrodefplural{DFS}[DFSs]{depth-first strategies}
\acrodefplural{ES}[ESs]{Eulerian strategies}
\acrodefplural{NE}[NE]{Nash equilibria}
\acrodefplural{UDFS}[UESs]{uniform depth-first strategies}
\acrodefplural{UES}[UESs]{uniform Eulerian strategies}

\title
{Search for an Immobile Hider on a Stochastic Network
}

\author%
{Tristan Garrec \\ 
CEREMADE, Universit\'e Paris-Dauphine \\
 place du Mar\'echal de Lattre de Tassigny, 75775 Paris Cedex 16, France\\ 
\& TSE-R, Toulouse School of Economics \\
Manufacture des Tabacs, 21 Allée de Brienne, 31015 Toulouse Cedex 6, France \\
\texttt{tristan.garrec@ut-capitole.fr}
\and
Marco Scarsini\\
Dipartimento di Economia e Finanza, LUISS \\
Viale Romania 32, 00197 Roma, Italy \\
\texttt{marco.scarsini@luiss.it}
}

\begin{document}

\maketitle

\begin{abstract}

Harry hides on an edge of a graph and does not move from there. 
Sally, starting from a known origin, tries to find him as soon as she can. 
Harry's goal is to be found as late as possible.
At any given time, each edge of the graph is either active or inactive, independently of the other edges, with a known probability of being active. 
This situation can be modeled as a zero-sum two-person stochastic game.
We show that the game has a value and we provide upper and lower bounds for this value.
Finally, by generalizing optimal strategies of the deterministic case, we provide more refined results for trees and Eulerian graphs.


\bigskip 


\bigskip


\end{abstract}

\acresetall

\maketitle

\newpage

\section{Introduction}\label{se:intro}

\subsection{The problem}\label{suse:problem}

	In a typical search game a hider hides in a space and a searcher, starting from a specified point, searches for the hider, trying to find him as fast as possible. 
	Often the space where the hider hides is assumed to be a network. 
	In almost all  existing versions of the game the network is fixed and all the edges are always available to the searcher.
	In real life it is often the case that some edges of the network are momentarily unavailable, for various reasons.  
	For instance when the police are looking for a suspect in a city, it is possible that the presence of traffic, or civilians, or other unexpected obstacles, forces them to deviate from the planned path.
	Most often the obstacles on the network are not permanent, but vary with time. 
	For instance, traffic may be intense in an area of the city at some time and in a different area at a different time. 
	The vehicles involved in an accident at some point get removed from the road and traffic goes back to normal.
	In a more common scenario, a road may be unavailable because of a red traffic light.
	This scenario, although simple to describe, would require considering stochastic dependence among the availability of different edges.
	
	Similar scenarios appear for instance when a rescue team is searching for miners in a mine. 
	Explosions or landslides may force the rescuers to change the course of actions. 
	Although in this case we do not have an adversarial hider, we can frame the situation as a zero-sum game, by considering the worst-case scenario, a game against Nature.
	
	It is clear that the stochastic elements that affect the shape of the network must be taken into account by both the hider and the searcher. 
	Consider the set of edges available to a searcher at a specific time.
	If the edge that she would have chosen is unavailable, she has two options: she can either wait until the edge becomes available, or she can take a different edge.
	Her choice clearly depends on the probability that each edge is available, on the structure of the network, and on her position in the game.

\subsection{Our contribution}\label{suse:contribution}

We study a hide-search model where a hider (Harry) hides on an edge of a graph and a searcher (Sally) travels around the graph in search of Harry. Her goal is to find him as soon as possible.

The novelty of the model is that, due to various circumstances, at any given time, some edges may be unavailable, so the graph randomly evolves over time.
At each stage, each edge $\edge$ of the  graph is, independently of the others, active with probability $\proba_{\edge}$ and inactive with probability $1-\proba_{\edge}$.
 
	At the beginning of the game, Harry hides on one edge of his choice and is immobile for the rest of the game.
	Starting from an initial vertex, called the root of the graph, Sally chooses at each stage a vertex among those reachable through active edges in the neighborhood of her current vertex. An equivalent approach is to let Sally choose an available incident edge, if any.
	The game ends when Sally traverses the edge where Harry is hidden, and his payoff is the number of  stages needed for the game to end. 
	So, Sally tries to minimize this time needed to find Harry and Harry aims to maximize this time. 
	This can be modeled as a zero-sum two-person game.

	We first examine the deterministic version of the game when  $\proba_{\edge}=1$ for each edge $\edge$. 	
	This game has a value and optimal strategies. 
	Analogously to well-known models in continuous time, we provide an upper and lower bound for this value, which correspond, for a fixed number of edges, to the value of games played on trees and on Eulerian graphs, respectively. 
	We also characterize optimal strategies when the graph is either a tree or an Eulerian graph. 
	We then turn to the stochastic framework and  show that, even in this case, the game has a value for all positive $\proba_{\edge}$. 
	We provide an upper and lower bound for this value and  show that it converges to the value of the deterministic game when $\proba_{\edge}\to 1$ for each edge $\edge$.
	We consider some particular instances when all $\proba_{\edge}$ are equal.
	We generalize optimal strategies of the deterministic setting to the stochastic one and obtain upper bounds on the value of the games played on binary trees and on parallel Eulerian graphs. 
	The upper bounds are tight when Sally is restricted to some search trajectories. 
	Finally we solve the stochastic search games played on the line and on the circle.
	
	We are aware that the stochastic representation of reality in our model is quite simplistic, but we see this as a first step to analyze search games under uncertainty on the network structure. 
	Moreover, even under our simplifying assumptions, the analysis is already quite complicated and general results are hard to achieve.

\subsection{Related literature}\label{suse:related-literature}

Several types of \acp{HSG} have been studied by various authors under different assumptions. 
\citet{Neu:CTG1953} studied a discrete version of the model where a hider hides in a cell $(i,j)$ of a matrix and a searcher chooses a row or column of the matrix; she finds the hider if the row or column contains the cell $(i,j)$. 
The problem was framed as a two-person zero-sum game.
Several variations of this discrete game were studied by various authors, among them 
\citet{Neu:JSIAM1963,Efr:JSIAM1964,GitRob:NRLQ1979,RobGit:NRLQ1978,Sak:JORSJ1973,Sub:JAP1981,BerMen:OR1986,BasBosRuc:JOTA1990}.

The search game with an immobile hider was introduced by  \citet{Isa:Wiley1996}.
\citet{BecNew:IJM1970} considered a continuous \ac{HSG} with a hider hiding on a line according to some distribution and a searcher, starting from an origin and moving at fixed speed, tries to find the hider as soon as possible. 
The continuous model was then generalized by \citet{Gal:IJM1972,Gal:SJAM1974,GalCha:SJAM1976}, who, among other things extended the state space from a line to a plane.

More relevantly to our paper, some authors dealt with \acp{HSG} on a network. Among them, \citet{Bos:SIAMJADM1984} studied a  discrete version of a continuous \ac{HSG} proposed by  \citet{Gal:AP1980}. This game is played on a parallel multi-graph with three edges that join two vertices $A$ and $B$ and the searcher, starting from $B$ has to find an immobile hider. The fact that the network has an odd number of parallel edges and, therefore, is not Eulerian makes the problem difficult to solve. 
\citet{Kik:JORSJ1990,Kik:JORSJ1991} considered a \ac{HSG} where the hider hides in one of $n$ cells on a straight line and the searcher incurs some traveling cost. 
\citet{AndAra:N1990} considered a \ac{HSG} on a network and framed the problem as an infinite-dimensional linear program.
\citet{Gal:SJCO1979,ReiPot:IJGT1993,Cao:EJOR1995,DagGal:N2008,Alp:N2008} examined \acp{HSG}  on trees, Eulerian networks, and some more general classes. 
\citet{Pav:NRL1995,Gal:IJGT2000,Kik:NRL2004,AlpBasGal:IJGT2008,AlpBasGal:N2009} extended the analysis to more general networks.
\citet{Alp:OR2011} considered a find-and-fetch game on a tree where the searcher has to find a hider on a network and can travel at speed $1$ to find him, and then has to return to the origin at a different speed.
\citet{AlpLid:OR2013,AlpLid:AOR2019}
replaced the usual pathwise search with what they call expanding search, where the searched area of a rooted network expands over different paths from the origin at different speeds chosen by the searcher, in such a way that the sum of the speeds is fixed.
\citet{AlpLid:OR2015} dealt with a situation where the searcher can choose one of two speeds to travel and can detect the hider, when passing in front of him, only if she travels at the lower speed.
\citet{Alp:OR2017} considered a model where the hider can hide anywhere in a network and the searcher has to entirely traverse an edge before being able to turn around. 
This constraints gives the problem a more combinatorial flavor.
\citet{Alp:TCSfrth} consider a search game where the hider is constrained to hide in a fixed subset of the whole network. 
If this subset is the set of the midpoints of all the edges, then the model becomes similar to the one we use here, where the hider hides on edges. 
Related to our stochastic model, \citet{BocKorRod:ESA2018} dealt with a search model on a graph, where randomness is induced by potentially unreliable advice, that is, with some fixed probability each node is faulty and points to the wrong neighbor. 
\citet{SteWer:DAM1997} studied the complexity of a \ac{HSG} on a graph when the hider hides on one of the nodes of the graph. 
\citet{JosBat:EJOR2008} proposed a heuristic algorithm to find a hider hidden uniformly at random on a network. 

In the \ac{HSG} studied by \citet{Alp:SJCO2010,AlpLid:MOR2014} the searcher moves on a network at a speed that depends on her location and direction. 
An intuitive link can be established between the speed variations considered in these two articles, and the expected time to cross some edges considered in the present article. 
In particular the \acl{BDFS} that we define and use in \cref{se:tree} is strictly related to the depth-first search defined in \citet{Alp:SJCO2010}.

This article should also be put into perspective with \citet{AlpLid:unpublished} which deals with the question of knowing when depth-first search is optimal. Our article brings some elements to answer this question in the stochastic setting.

In a forthcoming paper \citet{GlaClaLin:ORfrth} considered a search game where an object is hidden in one of many discrete locations and the searcher can use one of two search modes: a fast but inaccurate mode or a slow but accurate one. 
The reader is referred to the classical book by \citet{AlpGal:Kluwer2003} for an extended treatment of search games and to \citet{Hoh:JORSJ2016} for a recent survey of the relevant literature. 

To the best of our knowledge, the model where edges of a network are present only with some probability has not been studied before in the framework of search games, but is standard in other fields. For instance, it is at the foundations of the classical model of random graphs proposed by \citet{ErdRen:PMD1959,ErdRen:MTAMKIK1960,ErdRen:BIIS1961}, where, given a set of vertices, a random graph is generated by creating an edge between any two pairs of vertices independently with probability $\proba$. 
A similar model is studied in percolation theory, where edges of a graph are independently active with probability $\proba$ and one relevant problem is the number of clusters in the random graph and, as a consequence, the possibility of reaching one vertex starting from another one.
The reader is referred, for instance, to 
\citet{Gri:Springer1999,
Bol:CUP2001,Bol:CUP2006,
Hof:CUP2017} for a general treatment of random graphs and percolation.  \citet{BolKunLea:JCT2013} considered a cop and robbers games played on a random graph.
Some intriguing interactions between percolation and game theory have been recently studied by
\citet{DayFal:arXiv2018,HolMarMar:PTRFfrth}, who considered two-person zero-sum games on a graph with alternating moves.

\subsection{Organization of the paper}
The paper is organized as follows. 
\cref{se:model} describes the model.
\cref{se:deterministic} deals with the deterministic case, where all edges are active with probability $1$. 
\cref{se:value} shows existence of the value for the stochastic case and provides upper and lower bounds for this value.
\cref{se:dynamic-programming} uses dynamic programming to find best responses of the searcher against a known hiding distribution of the hider.
\cref{se:tree,se:eulerian} are devoted to the analysis of search games on trees and Eulerian graphs, respectively.
Most of the proofs can be found in \cref{se:appendix}.

\section{The model}\label{se:model}

\subsection{Notation}\label{suse:notation}

Given a finite set $A$, we call $\card A$ its cardinality and $\simplex(A)$ the set of probability measures on $A$. 

Let $\graph=(\vertices,\edges)$ be a connected undirected graph, where $\vertices$ is the nonempty finite set of vertices and $\edges$ is the nonempty finite set of edges. All edges have length $1$. The degree of a vertex $\vertex$ is denoted $\degr(\vertex)$.
There exists a special vertex $\source\in \vertices$, called the \emph{root} of the graph $\graph$. 
Let $\graphs$ be the set of subgraphs of $\graph$. For all $\vertex\in \vertices$, we call $\neighbor(\graph,\vertex) $ the \emph{immediate neighborhood} of $\vertex$ in $\graph$:
\begin{equation}\label{eq:neighbors}
\neighbor(\graph,\vertex) = \{\vertex\} \cup \{\vertexalt\in \vertices \vert \{\vertex,\vertexalt\}\in \edges\}.
\end{equation}

The graph will evolve in discrete time as follows.
Let $\probas=(\proba_{\edge})_{\edge\in\edges} \in(0,1]^\edges$. At each stage $\per\geq 1$, each edge $\edge\in\edges$ is active with probability $\proba_{\edge}$ or inactive with probability $1-\proba_{\edge}$, independently of the other edges.
This defines a random graph process on $\graphs$ denoted $(\graph_{\per})_{\per} = (\vertices,\edges_{\per})_{\per\geq 1}$, where $\edges_{\per}$ is the random set of active edges at time $\per$.

\subsection{The game}

We consider a stochastic zero-sum game $\game=\angles*{\graph,\source,\probas}$ with two players: a maximizer, called the \emph{hider} (Harry), and a minimizer, called the \emph{searcher} (Sally). We call this game a \acfi{SSG}\acused{SSG}.

The game is played as follows. 
At stage $0$ both players know $\graph_{0}=\graph$ and the initial position of the searcher $\vertex_{0} = \source$. 
The hider chooses an edge $\edge\in\edges$. 
Then the graph $\graph_{1}$ is drawn and the searcher chooses $\vertex_{1}\in \neighbor(\graph_{1},\vertex_{0})$. 
If $\{\vertex_{0},\vertex_{1}\}=\edge$, then the game ends and the payoff to the hider is $1$, otherwise the graph $\graph_{2}$ is drawn and the game continues. 
Inductively, at each stage $\per\geq 1$, knowing $\history_{\per}=(\graph_{0},\vertex_{0},\dots,\graph_{\per-1},\vertex_{\per-1},\graph_{\per})$, the searcher chooses $\vertex_{\per}\in \neighbor(\graph_{\per},\vertex_{\per-1})$. 
If $\{\vertex_{\per-1},\vertex_{\per}\}=\edge$, then the game ends and the payoff to the hider is $\per$, otherwise the graph $\graph_{\per+1}$ is drawn and the game continues.

Hence in this \ac{SSG}, the state space is $\graphs\times \vertices$, the action set of the hider is $\edges$, and the action set of the searcher in state $(\graph',\vertex)\in\graphs\times \vertices$ is $\neighbor(\graph',\vertex)$.
We now describe the sets of strategies of the players. For $\per\geq 0$, let $\histories_{\per} =  \graphs \times (\graphs \times \vertices)^{\per}$ be the set of histories at stage $\per$ and let $\histories = \bigcup_{\per\geq 0} \histories_{\per}$ be the set of all histories. 
Call $\pures$ the set of (behavior) strategies of the searcher, that is the strategies $\mixed\colon \histories\to \hull(\vertices)$ such that $\mixed(\history_{\per})\in\hull(\neighbor(\graph_{\per},\vertex_{\per-1}))$.

We call pure the strategies $\pure$ such that, for all $\per\geq 0$ and all $\history_{\per}\in \histories_{\per}$, 
\begin{equation*}
\pure(\history_{\per})=\vertex_{\per}\in\neighbor(\graph_{\per},\vertex_{\per-1}).
\end{equation*}

A behavior strategy $\mixed$ naturally induces a probability measure on each $\histories_{\per}$, for every $\per\ge 1$, which can be uniquely extended to  $\histories_{\infty}$ by Kolmogorov's extension theorem. 
This probability is denoted $\prob_\mixed$ and the corresponding expectation is denoted $\ex_\mixed$.

A mixed strategy of the searcher is a probability distribution over pure strategies, endowed with the product $\sigma$-algebra. By Kuhn's theorem, behavior and mixed strategies are equivalent \citep[see, e.g.,][]{Aum:AGT1964,Sor:Springer2002}.
The sets of pure and mixed strategies of the hider are  $\edges$ and $\hull(\edges)$, respectively. 
Pure strategies of the hider and the searcher will usually be denoted with the letters $\edge$ and $\pure$ respectively, while mixed and behavior strategies will usually be denoted with the letters $\mixedh$ and $\mixed$, respectively. We denote $\ud$ the \ac{UD} on $\edges$.

Finally, the payoff function of the hider is the function $\pay \colon \edges\times \pures \to \R_{+}\cup\{+\infty\}$, defined as
\begin{equation}\label{eq:payoff}
\pay(\edge,\mixed)=\ex_\mixed\bracks{\inf\braces{\per \ge 1 \vert \braces{\vertex_{\per-1},\vertex_{\per}}=\edge}},
\end{equation}
where the infimum over the empty set is $+\infty$. 
The function $\pay$ is linearly extended to $\hull(\edges)$.
The goal of the hider is thus to maximize the expected time by which he is found by the searcher, while the goal of the searcher is to minimize the expected time by which she finds the hider.

\section{Deterministic search games}
\label{se:deterministic}

\cref{pr:valueexistence} below will show that the search game $\angles{\graph,\source,\probas}$ has a value, which we  denote $\val(\probas)$. If $\proba_{\edge}$ is equal to $1$ for all $\edge\in\edges$, we then recover a search game with an immobile hider played on a graph. We call this game a \acfi{DSG}\acused{DSG}. \acp{DSG} have a value $\val(\boldsymbol{1})$.

We recall some important definitions and results for \acp{DSG}. Versions of these results are well known when the game is played in continuous time over a continuous network  \citep[see, e.g.,][]{AlpGal:Kluwer2003}.  

\begin{definition}\label{de:cycle-types}
\begin{enumerate}[(i)]
\item
A cycle in an graph is called \emph{Eulerian} if it uses each edge exactly once. If such a cycle exists, the graph is called Eulerian.

\item
A \emph{Chinese postman tour} is a cycle of minimal length that visits each edge. In Eulerian graphs, the Chinese postman tours are the Eulerian cycles.
\end{enumerate}
\end{definition}

\begin{definition}\label{de:chinese-postman-strategy}
\begin{enumerate}[(i)]
\item
The \acfi{UES}\acused{UES} is a mixed strategy that mixes over all Eulerian cycles with equal probability. 

\item
The \acfi{UCPS}\acused{UCPS} is a mixed strategy that mixes over all Chinese postman tours with equal probability. 
\end{enumerate}
\end{definition}

In \cref{de:chinese-postman-strategy} above, note that the \ac{UCPS} is not the same as a Random Chinese Postman Tour usually found in the literature. In a Random Chinese Postman Tour, the searcher follows equiprobably a Chinese postman tour or its reverse. Both strategies would be optimal in \cref{pr:tree} below, however only the \ac{UCPS} generalizes well to the stochastic setting.

When considering trees, we will endow them with an orientation outgoing from the root. This orientation does not affect the behavior of the searcher, who can traverse any edge in any direction, but is just needed to state and prove some of our results.

Let $\graph=\tree$ be a tree. If $\vertex$ is a vertex of $\tree$, then $\subtree_{\vertex}$ is the subtree that has $\vertex$ as a root and contains all edges below $\vertex$ in the original tree $\tree$. Hence $\tree=\subtree_{\source}$.

If $\edge$ is an edge of $\graph$, then $\subtree_{\edge}\coloneqq\{\edge\}\cup\subtree_{\vertex}$ where $\vertex$ is the head of $\edge$, i.e., $\subtree_{\edge}$ includes $\edge$ and the maximal subtree below the head of $\edge$. 
We denote $\edges_{\vertex}$ (resp. $\edges_{\edge}$) the  set of edges of $\subtree_{\vertex}$ (resp. $\subtree_{\edge}$).

The following definition is an adaptation to our framework of what \citet[Section~3.3]{AlpGal:Kluwer2003} have in the continuous setting.

\begin{definition}\label{de:EBD}
The \acfi{EBD}\acused{EBD} $\ebd$ of the hider is the unique distribution  on $\edges$  that is supported  on the leaf edges and, for every branching vertex $\vertex$ with outgoing edges $\edge_{1},\dots,\edge_{\nout}$, satisfies
\begin{equation}\label{eq:EBD}
\frac{\ebd(\edges_{\edge_{i}})}{\card \edges_{\edge_{i}}} = \frac{\ebd(\edges_{\edge_{1}})}{\card \edges_{\edge_{1}}}, \quad\text{for all }i\in \{1,\dots,\nout\}.
\end{equation}
\end{definition}

\begin{proposition}
\label{pr:tree}
Let $\game=(\vertices,\edges)$.
In a  \ac{DSG} $\game=\angles{\graph,\source,\boldsymbol{1}}$  we have
\begin{equation}\label{eq:val(1)}
\val(\boldsymbol{1}) \leq \card \edges.
\end{equation} 
Moreover, $\val(\boldsymbol{1}) = \card \edges$ if and only if $\graph$ is a tree.
In this case, the \ac{EBD} and the \ac{UCPS} are optimal strategies.
\end{proposition}

We first prove the following lemma.

\begin{lemma}
\label{le:CPC}
Let $\graph=(\vertices,\edges)$ be a connected graph. 
Any Chinese postman tour has length 
\begin{enumerate}[\upshape{(}i\upshape{)}]
\item $2\card \edges$ if $\graph$ is a tree,

\item at most $2\card \edges -2$ if $\graph$ is not tree.
\end{enumerate}
\end{lemma}

\begin{proof}
If $\graph$ is a tree the result follows by induction on $\card \edges$. 

Suppose now that $\graph$ is not a tree. We again proceed by induction on $\card \edges $. There exists an edge $\edge = \{\vertexalt,\vertex\}\in \edges$ such that $\graph' = (\vertices,\edges\setminus\{\edge\})$ is connected.

If $\graph'$ is a tree, we consider a Chinese postman tour $\cycle\in\graph'$ starting at $\vertexalt$, such that the subtree with root $\vertex$ is the last visited.
Once the vertex $\vertex$ is visited for the last time on $\cycle$, we replace the end of the cycle---which has already been visited---with $\edge$, going straight from $\vertex$ to $\vertexalt$.
This new cycle in $\graph$ has length  at most $2(\card \edges -1) + 1 - 1 = 2\card \edges -2$, since the length of the cycle in  $\graph'$ is $2(\card \edges -1)$, the length of $\edge$ is $1$, and the number of the edges not visited a second time is at least $1$.

If $\graph'$ is not a tree, then it admits a Chinese postman tour $\cycle$ with length at most $2(\card \edges -1)-2$. We now consider the cycle $\cyclealt\in\graph$ which starts at $\vertexalt$, goes back and forth on $\edge$ and then follows the cycle $\cycle$ on $\graph'$. This cycle has length  $2(\card \edges -1)-2 + 2 = 2\card \edges-2$.
\end{proof}

The proof of \cref{pr:tree} will make use of the following lemma, which refers to a model for continuous networks in continuous time. 
Let $\contnet$ be a continuous tree network, and suppose that the edges of $\contnet$ have integer length. 
Then $\contnet$ is mapped to a tree graph $\tree$ in the natural way. 
The \ac{UCPS} and the \ac{EBD} are defined in a similar way in $\tree$ and in $\contnet$, and are naturally mapped from the graph setting to the continuous network setting, and vice versa.

\begin{lemma}[\protect{\citet{Gal:SJCO1979}, \citet[][Theorem~3.21]{AlpGal:Kluwer2003}}] 
\label{le:Alpern-Gall-T3.21}
Let $\contnet$ be a continuous tree network with total length $\totlength$. Then
\begin{enumerate}[\upshape{(}i\upshape{)}]
\item
The \ac{UCPS} is an optimal search strategy.

\item
The \ac{EBD} is an optimal hiding strategy.

\item
$\val(\boldsymbol{1})=\totlength$.
\end{enumerate}

If the continuous network $\contnet$ with total length $\totlength$ is not a tree, then $\val(\boldsymbol{1})<\totlength$.
\end{lemma}

\begin{proof}[Proof of \cref{pr:tree}]
If $\graph$ is a tree, the result follows from  \cref{le:Alpern-Gall-T3.21}.
Indeed, in the discrete setting, hiding on edges that are not leaves is strictly dominated. Similarly in the continuous setting, hiding at a point of the tree which is not terminal is strictly dominated. Hence the \ac{UCPS} guarantees the value of the continuous game in the discrete one---with the natural mapping. Moreover, since the set of hiding strategies in the discrete setting is a subset of the set of hiding strategies on the continuous setting---again with the natural mapping---the \ac{EBD} guarantees in the discrete game the value of the continuous one.

If $\graph$ is not a tree, suppose that the searcher uniformly chooses between any Chinese postman tour, and let the hider choose an edge $\edge$.  For any fixed Chinese postman tour of length $n$, $\edge$ has position $k$ in the cycle and position $n-k+1$ in the reverse cycle. 
By \cref{le:CPC}, $n\le 2\card \edges -2$, hence, the payoff is at most 
\[
\frac{k + 2\card \edges -2 - k +1}{2} = \card \edges - \frac{1}{2} < \card \edges. \qedhere
\] 
\end{proof}

\begin{proposition}
\label{pr:eul}
Let $\graph=(\vertices,\edges)$.
In a  \ac{DSG}  $\game=\angles{\graph,\source,\boldsymbol{1}}$ we have
\begin{equation}\label{eq:Val(1)}
\val(\boldsymbol{1})\geq \frac{\card \edges +1}{2}.
\end{equation}
Moreover, if $\card \edges>1$, then 
\begin{equation}\label{eq:Val(1)=}
\val(\boldsymbol{1})=\frac{\card \edges +1}{2}.
\end{equation}
if and only if $\graph$ is Eulerian.
In this case, the \ac{UD} on $\edges$ and the \ac{UES} are optimal strategies.
\end{proposition}

\begin{proof}
Suppose the hider hides uniformly over $\edges$. Now let the searcher choose any sequence of edges (without necessarily following a path in $\graph$). Then if the searcher does not search the same edge twice during his $\card \edges$ first picks, the payoff is $(\card \edges +1)/2$, hence the lower bound.  
Suppose $\card{\edges} > 1$,  it is clear that this bound is reached only in Eulerian graphs, following an Eulerian cycle, because, if the graph is not Eulerian, then an edge is visited twice.
Finally, using an  argument similar to the one used in  \cref{pr:tree}, we can show that the uniform Eulerian strategy yields the payoff $(\card \edges +1)/2$ against any strategy of the hider.
\end{proof}

Together, \cref{pr:tree,pr:eul} yield the next theorem, whose continuous version is a cornerstone of the search game literature. 
It gives bounds on the value of deterministic search games played on any graphs. Moreover, it shows that Eulerian graphs and trees are the two extreme classes of graphs in term of value of the game. 

\begin{theorem}
For any graph $\graph=(\vertices,\edges)$, the value of the  \ac{DSG}  $\game=\angles{\graph,\source,\boldsymbol{1}}$ satisfies 
\begin{equation}\label{eq:bounds-Val(1)}
\frac{\card{\edges}+1}{2}\leq \val(\boldsymbol{1}) \leq \card{\edges}.
\end{equation}
Moreover, if $\card{\edges} > 1$, the upper bound is reached if and only if $\graph$ is a tree and the lower bound is reached if and only if $\graph$ is an Eulerian graph.

If $\graph$ is an Eulerian graph, then the \ac{UD} on $\edges$ and the \ac{UES}  are optimal strategies.

If $\graph$ is a  tree,  then the \ac{EBD} and the \ac{UCPS} are optimal strategies.
\end{theorem}

In  \cref{se:tree,se:eulerian} we focus on subclasses of these two extreme classes that are Eurelian graphs and trees. Both subclasses  have a recursive structure. We generalize the strategies of interest to our stochastic setting and derive bounds on the value. We also prove that these strategies are optimal in the cases of circles and lines.

\section{Value of the game}
\label{se:value}

\begin{proposition}
\label{pr:valueexistence}
For any $\probas\in(0,1]^\edges$ the  \ac{SSG} $\angles{\graph,\source,\probas}$ has a value $\val(\probas)$. Moreover both players have an optimal strategy.
\end{proposition}

The proof of \cref{pr:valueexistence} is postponed to \cref{se:appendix}.

\begin{proposition}
\label{pr:boundsonthevalue}
For all $\probas\in (0,1]^\edges$ the value of the  \ac{SSG} $\angles{\graph,\source,\probas}$ satisfies
\begin{equation}\label{eq:value-bounds}
\frac{\val(\boldsymbol{1})}{1-(1-\min_{\edge\in\edges} \proba_{\edge})^{\maxdegree}}\leq \val(\probas)\leq  \frac{\val(\boldsymbol{1})}{\min_{\edge\in\edges} \proba_{\edge}},
\end{equation}
where $\maxdegree$ is the maximum degree of $\graph$. 

As a consequence
\begin{equation}\label{eq:limit-value}
\val(\probas)\to\val(\boldsymbol{1}), \text{ as }\min_{\edge\in\edges}\proba_{\edge} \to 1.
\end{equation}
\end{proposition}

\begin{proof}
The hider guarantees the lower bound by playing as in the  \ac{DSG}. In expectation the searcher waits at least $(1-(1-\min_{\edge\in\edges} \proba_{\edge})^{\maxdegree})^{-1}$ for a neighbor edge to be active.

We map a strategy of the searcher in the  \ac{DSG} to the strategy in the  \ac{SSG} following the same path, even if it means waiting for an edge to be active. The searcher guarantees the upper bound since it takes in expectation at most $1/\min_{\edge\in\edges} \proba_{\edge}$ stages to cross a single edge.
\end{proof}

\section{Dynamic programming}
\label{se:dynamic-programming}

The next proposition is a dynamic programming formula which allows to find best responses of the searcher against a known hiding distribution of the hider. The activation parameters $\probas\in(0,1]^\edges$ are fixed and we omit them.

For all $\graph_{1}\in \graphs$, $\vertex_{0}\in \vertices$, $\subedges\subset \edges$ and $\mixedh\in \simplex(\subedges)$, we define
\begin{equation}\label{eq:Value}
\Val(\graph_{1},\vertex_{0},\subedges,\mixedh) = \min_{\pure\in\pures} \ex_{\pure}\left[\sum_{\edge\in \subedges}\mixedh(\edge)\inf\{\per\geq 1 \vert \{\vertex_{\per-1},\vertex_{\per}\} = \edge\}\right].
\end{equation}
This quantity represents the value of the (one player) game in which the searcher knows  the graph $\graph_{1}$ and the distribution $\mixedh$ of the hider on $\subedges\subset \edges$, starts from $\vertex_{0}$ and chooses immediately $\vertex_{1}\in \neighbor(\graph_{1},\vertex_{0})$ at the first stage, before $\graph_{2}$ is drawn (and then the game continues). In other words, in the true game, a graph $\graph_{1}$ is drawn before Sally starts playing. Here the graph $\graph_{1}$ is already fixed and Sally starts playing immediately.
\begin{proposition}
If $\subedges = \varnothing$, then $\Val(\graph_{1},\vertex_{0},\subedges,\mixedh)=0$. Otherwise
\begin{equation}
\label{eq:dyn-prog}
\Val(\graph_{1},\vertex_{0},\subedges,\mixedh) =\\
 1+ \min_{\vertex_{1}\in \neighbor(\graph_{1},\vertex_{0})} \mixedh(\subedges\setminus \{\vertex_{0},\vertex_{1}\}) \ex\left[\Val\left(\graph_{2},\vertex_{1},\subedges\setminus\{\vertex_{0},\vertex_{1}\},\mixedh^{\{\vertex_{0},\vertex_{1}\}}\right)\right],
\end{equation}
where $\mixedh^{\{\vertex_{0},\vertex_{1}\}}(\cdot) = \frac{1}{\mixedh(\subedges\setminus \{\vertex_{0},\vertex_{1}\})}\mixedh(\cdot)$, and the randomness in  \cref{eq:dyn-prog} is over $\graph_{2}$.
\end{proposition}
\begin{proof}
If the searcher finds the hider in the first stage, which happens with probability $\mixedh(\{\vertex_{0},\vertex_{1}\})$, then the game ends and the continuation payoff is $0$. 
On the other hand, if the searcher does not find the hider in the first stage, which happens with probability $1-\mixedh(\{\vertex_{0},\vertex_{1}\})$, then the game continues with continuation payoff 
\begin{equation}\label{eq:cont-payoff}
\ex\left[\Val\left(\graph_{2},\vertex_{1},\subedges\setminus\{\vertex_{0},\vertex_{1}\},\mixedh^{\{\vertex_{0},\vertex_{1}\}}\right)\right],
\end{equation} 
since the edge $\{\vertex_{0},\vertex_{1}\}$ has been visited and the next graph $\graph_{2}$ is yet to be drawn.
\end{proof}

\section{Stochastic search games on trees}
\label{se:tree}

In this section and in the following one we assume
\begin{equation}\label{eq:pe-equal}
\proba_{\edge}=\proba\in(0,1],\quad\text{for all }\edge\in\edges.
\end{equation}
Moreover in this section we assume that $\graph$ is a tree $\tree$ with origin $\source$. 
Remark that in a tree, any strategy of the hider that consists in hiding in edges other than leaf edges is strictly dominated.

\subsection{Depth-first strategies and the equal branching density}

We define a particular class of strategies of the searcher in trees, called depth-first strategies. They have the property of never going backward at a vertex before having visited the whole subtree. They generalize the Chinese postman tours of the deterministic setting.

\begin{definition}
A \acfi{DFS}\acused{DFS} on a tree is a strategy of the searcher that prescribes the following, when arriving at a vertex:
\begin{itemize}
\item 
if the set of un-searched and active outgoing edges is non-empty, take one of its edges (possibly at random);
\item 
if all the un-searched outgoing edges are inactive, wait;
\item 
if all outgoing edges have been searched and the backward edge is active, take it;
\item 
if all outgoing edges have been searched and the backward edge is inactive, wait.
\end{itemize}
The \acfi{UDFS}\acused{UDFS} is the \ac{DFS} that, at every vertex, randomizes uniformly between all active and un-searched outgoing edges.
\end{definition}

\begin{definition}
A \ac{DFS} on $\tree$ induces an expected time to travel from the origin $\source$ back to it, covering the entire tree. 
This is called the \emph{cycle time} of $\tree$ and is denoted $\cycletime(\source)$.  
For any vertex or edge $\vertexedge$, the cycle time of $\subtree_{\vertexedge}$ is denoted $\cycletime(\vertexedge)$.
\end{definition}

Notice that $\cycletime(\source)$ depends on $\probas$, but is independent of the choice of \ac{DFS}.

We now generalize \cref{de:EBD} to the stochastic setting, where the relevant quantity is not the number of edges of the subtrees, but rather their cycle times.

\begin{definition}\label{de:EBD-p}
The \acfi{EBD}\acused{EBD} $\ebd$ of the hider is the unique distribution on the leaf edges such that, for every branching vertex $\vertex$ with outgoing edges $\edge_{1},\dots,\edge_n$, we have
\begin{equation}\label{eq:EBD-p}
\frac{\ebd(\edges_{\edge_{i}})}{\cycletime(\edge_{i})} = \frac{\ebd(\edges_{\edge_{1}})}{\cycletime(\edge_{1})}, \quad\text{for all }i\in \{1,\dots,n\}.
\end{equation}
\end{definition}
Notice that \cref{de:EBD,de:EBD-p} coincide when $\proba_{\edge}=1$ for all $\edge\in\edges$.

\subsection{Binary trees}
\subsubsection{Generalities}

In these sections we consider games played on binary trees, i.e., trees with at most two outgoing edges at any vertex.
We call $\trees$ the set of binary trees. 
\acp{DFS} allow us to obtain an upper bound for the value, when $\proba$ is large enough. 
We also prove that this upper bound is the value of the game in which Sally is restricted to play  \acp{DFS}.
As a by-product we will show that, for every $\proba\in(0,1]$, the \ac{UDFS} and \ac{EBD} are a pair of optimal strategies when the game is played on a line.

\begin{definition}\label{de:UD}
Given a tree $\tree=(\vertices,\edges)$, we define the function $\diffedge\colon\trees\to\R$ recursively as follows, where, for the sake of simplicity we use the notations $\diffedge(\edge)=\diffedge(\tree_\edge)$ and $\diffedge(\vertex)=\diffedge(\tree_\vertex)$:

If $\tree$ has a single edge $\edge=(\source,\vertex)$, as in \cref{fig:oneedge}, then 
\begin{equation}\label{eq:Lambda-single-edge}
\diffedge(\source)=\diffedge(\edge)=\diffedge(\vertex)=0.
\end{equation}

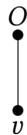
\begin{figure}[!ht]
  \begin{center}
    \begin{tikzpicture}[scale=1]
\node[above] at (0,0){$\source$};
\node[below] at (0,-1){$\vertex$};

\node at (0,0){$\bullet$};
\node at (0,-1){$\bullet$};

\draw[-] (0,0) to (0,-1);
\end{tikzpicture}
    \caption{One edge}
    \label{fig:oneedge}
  \end{center}
\end{figure}  
  
If $\degr(\source)=1$ and $\edge=(\source,\vertex)$, as in \cref{fig:Odegree1}, then $\diffedge(\source)=\diffedge(\edge)=\diffedge(\vertex)$.
  
\begin{figure}[!ht]
  \begin{center}
    \begin{tikzpicture}[scale=1]
\node[above] at (0,0){$\source$};
\node[below] at (0,-1){$\vertex$};
\node[left] at (0,-0.5){$\edge$};

\node at (0,0){$\bullet$};
\node at (0,-1){$\bullet$};
\node at (0,-2){$\subtree_{\vertex}$};

\draw[-] (0,0) to (0,-1);
\draw [line width=.5pt,dash pattern=on 1pt off 2pt] (0,-2) circle(1cm);
\end{tikzpicture}
    \caption{$\source$ has degree $1$}
    \label{fig:Odegree1}
  \end{center}
\end{figure}  

If $\tree$ has two edges and $\degr(\source)=2$, as in \cref{fig:twoedges}, then 
\begin{equation}\label{eq:Lambda-2edges}
\diffedge(\source) = \frac{1}{2}\left(\frac{1}{1-(1-\proba)^{2}}-\frac{1}{\proba}\right).
\end{equation}
  
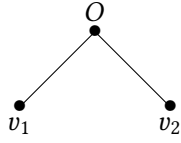
\begin{figure}[!ht]
  \begin{center}
    \begin{tikzpicture}[scale=1]
\node[above] at (0,0){$\source$};
\node[below] at (-1,-1){$\vertex_{1}$};
\node[below] at (1,-1){$\vertex_{2}$};

\node at (0,0){$\bullet$};
\node at (-1,-1){$\bullet$};
\node at (1,-1){$\bullet$};

\draw[-] (0,0) to (-1,-1);
\draw[-] (0,0) to (1,-1);
\end{tikzpicture}
    \caption{Two edges}
    \label{fig:twoedges}
  \end{center}
\end{figure}

If $\degr(\source)=2$,  $\edge_{1}=(\source,\vertex_{1})$, and $\edge_{2}=(\source,\vertex_{2})$, as in \cref{fig:Odegree2}, then 
\begin{equation}\label{eq:O-degree-2}
\diffedge(\source) = \frac{\cycletime(\vertex_{1})}{\cycletime(\vertex_{1})+\cycletime(\vertex_{2})}\diffedge(\vertex_{1})+\frac{\cycletime(\vertex_{2})}{\cycletime(\vertex_{1})+\cycletime(\vertex_{2})}\diffedge(\vertex_{2})+\frac{1}{2}\left(\frac{1}{1-(1-\proba)^{2}}-\frac{1}{\proba}\right).
\end{equation}

\begin{figure}[!ht]
  \begin{center}
    \begin{tikzpicture}[scale=1]
\node[above] at (0,0){$\source$};
\node[below] at (-1,-1){$\vertex_{1}$};
\node[below] at (1,-1){$\vertex_{2}$};
\node[left] at (-0.5,-0.5){$\edge_{1}$};
\node[right] at (0.5,-0.5){$\edge_{2}$};

\node at (0,0){$\bullet$};
\node at (-1,-1){$\bullet$};
\node at (1,-1){$\bullet$};
\node at (-1.25,-2){$\subtree_{\vertex_{1}}$};
\node at (1.25,-2){$\subtree_{\vertex_{2}}$};

\draw[-] (0,0) to (-1,-1);
\draw[-] (0,0) to (1,-1);
\draw [line width=.5pt,dash pattern=on 1pt off 2pt] (-1.25,-2) circle(1cm);
\draw [line width=.5pt,dash pattern=on 1pt off 2pt] (1.25,-2) circle(1cm);
\end{tikzpicture}
    \caption{$\source$ has degree $2$}
    \label{fig:Odegree2}
  \end{center}
\end{figure}
\end{definition}

The function $\diffedge$  depends on $\proba$, but we do not make the dependence explicit.

\begin{lemma}\label{le:delta} 
Let $\vertex$ be a branching vertex with outgoing edges $\edge_{1}$ and $\edge_{2}$. Then for all $\proba\in(0,1]$,
\[
\frac{\vert\diffedge(\edge_{1})\vert + \vert\diffedge(\edge_{2})\vert}{\cycletime(\edge_{1})+\cycletime(\edge_{2})}<\frac{1}{2}.
\]
\end{lemma}

The proof of \cref{le:delta} is postponed to \cref{se:appendix}.
We now define the biased depth-first (behavior) strategy of the searcher.

\begin{definition}
\label{def:bdfs}
Assume that vertex $\vertex$ has outgoing edges $\edge_{1}$ and $\edge_{2}$ and they are both active and un-searched. 
A \ac{DFS} strategy $\mixed_{\bdfs}$ is called the \acfi{BDFS}\acused{BDFS} if it takes $\edge_{1}$ with probability $\bdfs(\edge_{1})$ and $\edge_{2}$ with probability $\bdfs(\edge_{2})$, where
\begin{align}\label{eq:bdsf}
\bdfs(\edge_{1}) &= \proj_{[0,1]}\left( \frac{1}{2} + \frac{\diffedge(\edge_{1})-\diffedge(\edge_{2})}{\cycletime(\edge_{1})+\cycletime(\edge_{2})} \frac{1-(1-\proba)^{2}}{\proba^{2}}\right)\\
\bdfs(\edge_{2}) &=  1- \bdfs(\edge_{1}),
\end{align}
where $\proj_{[0,1]}$ indicates the projection on $[0,1]$. 
\end{definition}

\begin{theorem}
\label{th:BQDFequalizingleafs}
There exists $\proba_{0}\in(0,1)$ such that  for all $\proba\geq \proba_{0}$, the time to reach any leaf edge using the \ac{BDFS} is $\frac{1}{2}\cycletime(\source) + \diffedge(\source)$. 
Hence for all $\proba\geq \proba_{0}$, we have
\begin{equation}\label{eq:valp-less}
\val(\proba)\leq \frac{1}{2}\cycletime(\source) + \diffedge(\source).
\end{equation}
\end{theorem}

The proof of \cref{th:BQDFequalizingleafs} is postponed to \cref{se:appendix}.

\begin{theorem}
\label{th:EBDequalizingQDF}
The \ac{EBD} of the hider yields the same payoff against any \ac{DFS} of the searcher, and this payoff is $\frac{1}{2}\cycletime(\source) + \diffedge(\source)$.
\end{theorem}

The proof of \cref{th:EBDequalizingQDF} is postponed to \cref{se:appendix}.

Note that \cref{def:bdfs,th:BQDFequalizingleafs,th:EBDequalizingQDF} above have a superficial resemblance to results on the value and on biased depth-first strategies in \citep{Alp:SJCO2010,AlpLid:MOR2014}, where the searcher moves on a network at a speed that depends on her location and direction.

\cref{th:BQDFequalizingleafs,th:EBDequalizingQDF} imply that in a binary tree $\graph$, if \acp{DFS} are best responses to the \ac{EBD}, then there exists $\proba_{0}\in (0,1)$ such that for all $\proba\geq \proba_{0}$ the value of the game is $\frac{1}{2}\cycletime(\source)+\diffedge(\source)$. Moreover the \ac{BDFS} and the \ac{EBD} are optimal.

\cref{ex:no-best-response} below is an important counterexample, as it refutes the conjecture that \acp{DFS} are best responses to the \ac{EBD}.

\begin{example}\label{ex:no-best-response}
We study the game played on the tree represented in  \cref{fig:DFnotbestresponsetoEBD}.

\begin{figure}[h]
\begin{center}
\begin{tikzpicture}[scale=1]
\node[above] at (0,0){$\source$};
\node[left] at (-2.5,-2.5){$\vertex_{1}$};
\node[right] at (1,-1){$\vertex_{2}$};
\node[below] at (0.5,-2){$\vertex_{21}$};
\node[below] at (1.75,-3){$\vertex_{22}$};

\node at (0,0){$\bullet$};
\node at (-0.5,-0.5){$\bullet$};
\node at (-1,-1){$\bullet$};
\node at (-1.5,-1.5){$\bullet$};
\node at (-2,-2){$\bullet$};
\node at (-2.5,-2.5){$\bullet$};
\node at (1,-1){$\bullet$};
\node at (0.5,-2){$\bullet$};
\node at (1.5,-2){$\bullet$};
\node at (1.75,-3){$\bullet$};

\node[left,above] at (-0.5,-0.5){$\edge_{1}$};
\node[right,above] at (0.5,-0.5){$\edge_{2}$};
\node at (1.5,-1.5){$\edge_{22}$};
\node at (0.5,-1.5){$\edge_{21}$};

\node at (1.85,-2.5){$\edge_{22}'$};
\node at (-2.5,-2.1){$\edge_{1}'$};

\draw[-] (0,0) to (-2.5,-2.5);
\draw[-] (0,0) to (1,-1);
\draw[-] (1,-1) to (0.5,-2);
\draw[-] (1,-1) to (1.5,-2);
\draw[-] (1.5,-2) to (1.75,-3);
\end{tikzpicture}
\end{center}
\caption{A counter-example}
\label{fig:DFnotbestresponsetoEBD}
\end{figure}
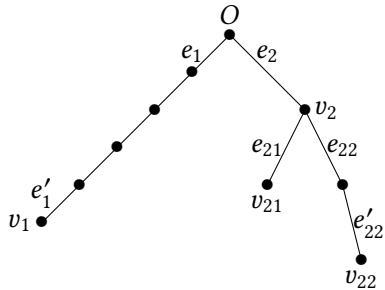

Consider the case where Sally visits $\vertex_{22}$ before any other leaf vertex. 
When she plays a \ac{DFS}, this event has positive probability.  
Assume also that, when she has returned to $\vertex_{2}$, after visiting $\vertex_{22}$, the edge $\edge_{2}$ is active but $\edge_{21}$ is not. 
At this point she can either take edge $\edge_{2}$ and visit $\vertex_{1}$ before $\vertex_{21}$  or wait until $\edge_{21}$ becomes active and visit $\vertex_{21}$ before $\vertex_{1}$. The first choice yields a lower payoff to Sally.

Indeed, visiting $\vertex_{1}$ first yields the continuation payoff
\begin{align*}
\pay_{1} =	\ebd(\edge_{1}')\left(1+ \frac{5}{\proba}\right) + \ebd(\edge_{21})\left(1 + \frac{12}{\proba}\right),
\end{align*} 
whereas visiting $\vertex_{21}$ first yields the continuation payoff
\begin{align*}
\pay_{2} =	\ebd(\edge_{21})\left(1+ \frac{1}{\proba}\right) + \ebd(\edge_{1}')\left(1 + \frac{8}{\proba}\right).
\end{align*}
The sign $\pay_{1} - \pay_{2}$ is the same as the sign of $11 \ebd(\edge_{21}) - 3 \ebd(\edge_{1}')$, which is the same as 
\begin{equation*}
\frac{11}{3}\left(\frac{7}{\proba}+\frac{1}{1-(1-\proba)^{2}}\right)- \frac{30}{\proba},
\end{equation*} 
which is negative for all $\proba\in(0,1)$.
\end{example}

\subsubsection{A simple binary tree}
\label{subsec_simplebinarytree}
We now present a game played on a tree (\cref{fig:completeresolution}) for which we give the value and a pair of optimal strategies for any value of $\proba\in(0,1]$.

\begin{figure}[ht]
\begin{center}
\begin{tikzpicture}[scale=1]
\node[above] at (0,0){$\source$};
\node[left] at (-1,-1){$\vertex$};
\node[right] at (1,-1){$\vertex_{2}$};
\node[below] at (0.5,-2){$z$};
\node[below] at (1.5,-2){$t$};

\node at (0,0){$\bullet$};
\node at (-1,-1){$\bullet$};
\node at (1,-1){$\bullet$};
\node at (0.5,-2){$\bullet$};
\node at (1.5,-2){$\bullet$};

\node[left,above] at (-0.5,-0.5){$\edge_{1}$};
\node[right,above] at (0.5,-0.5){$\edge_{2}$};
\node at (1.5,-1.5){$\edge_{22}$};
\node at (0.5,-1.5){$\edge_{21}$};

\draw[-] (0,0) to (-1,-1);
\draw[-] (0,0) to (1,-1);
\draw[-] (1,-1) to (0.5,-2);
\draw[-] (1,-1) to (1.5,-2);
\end{tikzpicture}
\end{center}
\caption{A simple binary tree}
\label{fig:completeresolution}
\end{figure}
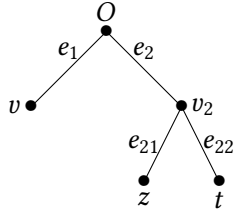
Let 
\begin{equation}\label{eq:p0}
\proba_{0} = \frac{9-\sqrt{65}}{8}\approx 0.12.
\end{equation} 

\paragraph{First case $\proba\geq \proba_{0}$:}
in this case, Sally's \ac{BDFS} and Harry's \ac{EBD} are a pair of optimal strategies. The value of the game is thus  
\[
\val(\proba) = \frac{1}{2}\cycletime(\source)+\diffedge(\source) = \frac{92-75\proba+15\proba^{2}}{\proba(15-7\proba)(2-\proba)}.
\]

\paragraph{Second case $\proba\leq \proba_{0}$:} 
Harry's strategy $\left(\frac{1}{3},\frac{1}{3},\frac{1}{3}\right)$  is optimal. We now describe an optimal strategy of Sally.
\begin{itemize}
\item If no  leaf edges have been visited:
\begin{itemize}
\item 
At $\source$: if $\edge_{1}$ is active, take it. Otherwise, if $\edge_{2}$ is active but  $\edge_{1}$ is not, take $\edge_{2}$.
\item 
At $\vertex_{2}$: take the first active edge between $\edge_{21}$ and $\edge_{22}$, drawing uniformly, if they both are.
\end{itemize}
\item 
If only $\edge_{1}$ has been visited, play the \ac{UDFS} in the continuation game.
\item 
If only $\edge_{21}$ (resp. $\edge_{22}$) has been visited, at $\vertex_{2}$:
\begin{itemize}
\item 
If $\edge_{22}$ (resp. $\edge_{21}$) is active, take it.
\item 
If $\edge_{2}$ is active but  $\edge_{22}$ (resp. $\edge_{21}$) is not, randomize, waiting at $\vertex_{2}$ with probability $\probbin(\proba)$ and taking $\edge_{22}$ (resp. $\edge_{21}$) with probability $1-\probbin(\proba)$.
\end{itemize}
\item 
If two leaf edges have been visited, go to the third leaf edge as quickly as possible.
\end{itemize}
The waiting probability $\probbin(\proba)$ is given by 
\[
\probbin(\proba) = \frac{8(2-\proba)-(1-\proba)(1+\proba)(2-\proba)}{8(2-\proba)(1-\proba)-\proba(1-\proba)^{2}}.
\]
The value of the game is 
\[
\val(\proba) = \frac{1}{3}\frac{37-33\proba+7\proba^{2}}{\proba(2-\proba)^{2}}.
\]

\subsubsection{The line}

We consider a \ac{SSG} played on a line. If the origin $\source$ is an extreme vertex, then the value of the game is $\card (\edges)/\proba$. We now suppose that the origin $\source$ is not an extreme vertex, and that the line has $\totedges=\length_{1}+\length_{2}$ edges ($\length_{1}$ on the left side of $\source$ and $\length_{2}$ on the right side) as shown in  \cref{fig:line}. 

\begin{figure}[ht]
\begin{center}
\begin{tikzpicture}[scale=0.7]
\node[above] at (0,0){$\source$};
\node[below,right] at (-2.5,-2.65){$\edge_{1}$};
\node[below,left] at (1.6,-1.6){$\edge_{2}$};

\node at (0,0){$\bullet$};
\node at (-1,-1){$\bullet$};
\node at (-2,-2){$\bullet$};
\node at (-3,-3){$\bullet$};
\node at (2,-2){$\bullet$};
\node at (1,-1){$\bullet$};
\node at (2,-2){$\bullet$};

\draw[-] (0,0) to (-3,-3);
\draw[-] (0,0) to (2,-2);
\end{tikzpicture}
\end{center}
\caption{The line with $\length_{1} = 3$ and $\length_{2}=2$}
\label{fig:line}
\end{figure}
In this case, for all $\proba\in (0,1]$ the \ac{BDFS} is the \ac{UDFS} $\eulerstr$, and the \ac{EBD} of the hider is 
\begin{equation*}
\ebd=\left(\frac{\length_{1}}{\totedges},\frac{\length_{2}}{\totedges}\right).
\end{equation*}

\begin{proposition}
\label{pr:line}
If the graph $\graph$ is a line, then \ac{DFS} are best responses to the \ac{EBD}. Hence, $(\ebd,\eulerstr)$ is a pair of optimal strategies.
\end{proposition}

\begin{proof}
Harry plays $\ebd$. At $\source$, whatever active edge Sally takes, the continuation payoff is
$(\length_{1}+\length_{2}-1)/\proba$.
Hence she does not profit from waiting for one specific edge to be active.
\end{proof}
Together with  \cref{th:BQDFequalizingleafs,th:EBDequalizingQDF},  \cref{pr:line} yields the following corollary.

\begin{corollary}
The value of the game played on the line with $\totedges$ edges is 
\[
\val(\proba)=\frac{1}{2}\cycletime(\source) + \diffedge(\source) = \frac{\totedges}{\proba}+\frac{1}{1-(1-\proba)^{2}}-\frac{1}{\proba},
\]
for all $\proba\in(0,1]$, if the root is not an extreme vertex. Moreover the \ac{EBD} and the \ac{UDFS}  are optimal strategies.
\end{corollary}

\section{Stochastic search games on Eulerian graphs}
\label{se:eulerian}

\subsection{Eulerian strategies and the uniform density}

For Eulerian graphs we define a strategy of the searcher, called  \acfi{ES}\acused{ES}, which generalizes an Eulerian cycle of the deterministic setting. 
At any vertex  an \ac{ES} chooses an active outgoing edge that had not previously been visited in  such a way that  the induced path is an Eulerian cycle. 
The \ac{ES} that at any vertex randomizes uniformly over the outgoing edges is called the \acfi{UES}\acused{UES} and is denoted $\eulerstr$.
\begin{definition}\label{de:cycle-time-EG}
The \ac{UES} on a Eulerian graph $\graph$ induces an expected time to travel from the origin $\source$ covering the entire Eulerian graph. 
This is called the \emph{cycle time} of $\graph$ and is denoted $\eulertime(\graph)$. 
\end{definition}

\subsection{Parallel Eulerian graphs}

\subsubsection{Generalities}

We call \emph{parallel graph} a graph where parallel paths link two vertices, one of these two vertices being the root $\source$, as in  \cref{fig:parallel}. Such a graph is denoted $\euler_{\npaths}(\lengthprof)$, where $\lengthprof=\parens{\length_{1},\dots,\length_{\npaths}}$ is the vector of the lengths of the parallel paths. The parallel uniform strategy of Sally consists in choosing at $\source$ uniformly between active and unsearched edges and then going straight to $\sink$ on the current parallel path (and similarly at $\sink$).

Remark that if the number of parallel paths $\npaths = 2\halfpaths$ is even, then the parallel graph is Eulerian and we call it a \emph{parallel Eulerian graph}. In this case, the parallel uniform strategy is the \ac{UES}.
For a parallel Eulerian graph $\euler_{2\halfpaths}(\lengthprof)$ with $2\halfpaths$ parallel lines, the cycle time of $\euler_{2\halfpaths}(\lengthprof)$ is
\begin{align*}
 \eulertime(\euler_{2\halfpaths}(\lengthprof)) = \sum_{\ppathalt=1}^{2\halfpaths} \left(\frac{1}{1-(1-\proba)^{\ppathalt}} + \frac{\length_\ppathalt-1}{\proba}\right).
\end{align*}

The \ac{UES} allows us to obtain an upper bound for the value. 
We also prove that this upper bound is the value of the game in which Sally is restricted to play  \acp{ES}.
As a by-product we will show that, for every $\proba\in(0,1]$, the \ac{UES} and \ac{UD} are a pair of optimal strategies when the game is played on a circle.

\begin{figure}[ht]
\begin{center}
\begin{tikzpicture}[scale=1]
\node[left] at (0,0){$\source$};
\node[right] at (5,0){$\sink$};

\node at (0,0){$\bullet$};
\node at (0.7,0.7){$\bullet$};
\node at (1.5,1.3){$\bullet$};
\node at (4.3,0.7){$\bullet$};
\node at (3.5,1.3){$\bullet$};
\node at (5,0){$\bullet$};
\draw[-] (0,0) to (0.7,0.7);
\draw[-] (0.7,0.7) to (1.5,1.3);
\draw[-] (5,0) to (4.3,0.7);
\draw[-] (3.5,1.3) to (4.3,0.7);
\draw[dotted] (3.5,1.3) to (1.5,1.3);

\node at (0.7,-0.7){$\bullet$};
\node at (1.5,-1.3){$\bullet$};
\node at (4.3,-0.7){$\bullet$};
\node at (3.5,-1.3){$\bullet$};
\draw[-] (0,0) to (0.7,-0.7);
\draw[-] (0.7,-0.7) to (1.5,-1.3);
\draw[-] (5,0) to (4.3,-0.7);
\draw[-] (3.5,-1.3) to (4.3,-0.7);
\draw[dotted] (3.5,-1.3) to (1.5,-1.3);

\node at (0.9,-0.4){$\bullet$};
\node at (1.85,-0.7){$\bullet$};
\node at (4.1,-0.4){$\bullet$};
\node at (3.15,-0.7){$\bullet$};
\draw[-] (0,0) to (0.9,-0.4);
\draw[-] (0.9,-0.4) to (1.85,-0.7);
\draw[dotted] (3.15,-0.7) to (1.85,-0.7);
\draw[-] (3.15,-0.7) to (4.1,-0.4);
\draw[-] (4.1,-0.4) to (5,0);
\draw[dotted] (3.5,-1.3) to (1.5,-1.3);

\node at (0.9,0.4){$\bullet$};
\node at (1.85,0.7){$\bullet$};
\node at (4.1,0.4){$\bullet$};
\node at (3.15,0.7){$\bullet$};
\draw[-] (0,0) to (0.9,0.4);
\draw[-] (0.9,0.4) to (1.85,0.7);
\draw[dotted] (3.15,0.7) to (1.85,0.7);
\draw[-] (3.15,0.7) to (4.1,0.4);
\draw[-] (4.1,0.4) to (5,0);
\draw[dotted] (3.5,1.3) to (1.5,1.3);

\end{tikzpicture}
\end{center}
\caption{A parallel Eulerian graph}
\label{fig:parallel}
\end{figure}
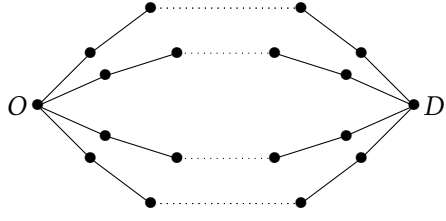

\begin{definition}
Given a parallel Eulerian graph $\euler_{2\halfpaths}(\lengthprof)$ with $2\halfpaths$ parallel lines, let $\altdiff_{\halfpaths}$ be the following quantity defined recursively:
\begin{equation}
\altdiff_{1} = \frac{1}{2} \left(\frac{1}{1-(1-\proba)^2}-\frac{1}{\proba}\right),
\end{equation}
 and for each $\halfpaths > 1$,
\begin{align*}
\altdiff_{\halfpaths} 
&= \frac{1}{2}\frac{1}{1-(1-\proba)^{2\halfpaths}} + \left(\frac{1}{2}-\frac{1}{2\halfpaths}\right)\frac{1}{1-(1-\proba)^{2\halfpaths-1}} \\
&\quad -\frac{1}{2\halfpaths}\left(\sum_{\ppathalt=1}^{2(\halfpaths-1)} \frac{1}{1-(1-\proba)^\ppathalt}+ \frac{1}{\proba}\right) + \frac{\halfpaths-1}{\halfpaths}\altdiff_{\halfpaths-1}.
\end{align*}
\end{definition}
Remark that $\altdiff_{\halfpaths}$ only depends on the number of parallel paths and not on their length.

\begin{theorem}
\label{th:EUequalizingedges}
On a parallel Eulerian graph $\euler_{2\halfpaths}(\lengthprof)$, the expected time to reach any edge using the \ac{UES} is 
\begin{equation}\label{eq:time-Euler}
\frac{\eulertime(\euler_{2\halfpaths}(\lengthprof))+\proba^{-1}}{2} + \altdiff_\halfpaths.
\end{equation} 
Hence, for all $\proba\in (0,1]$, we have 
\begin{equation}\label{eq:val-Euler}
\val(\proba)\leq \frac{\eulertime(\euler_{2\halfpaths}(\lengthprof))+\proba^{-1}}{2} + \altdiff_\halfpaths.
\end{equation}
\end{theorem}

The proof of \cref{th:EUequalizingedges} is postponed to \cref{se:appendix}.

\begin{theorem}
\label{th:UequalizingEU}
On a parallel Eulerian graph $\euler_{2\halfpaths}(\lengthprof)$, the uniform density of the hider yields the same payoff 
\begin{equation*}
\frac{\eulertime(\euler_{2\halfpaths}(\lengthprof))+\proba^{-1}}{2} + \altdiff_\halfpaths
\end{equation*}
against any Eulerian strategy of the searcher.
\end{theorem}

The proof of \cref{th:UequalizingEU} is postponed to \cref{se:appendix}.
\cref{th:EUequalizingedges,th:UequalizingEU} imply that in a parallel Eulerian graph $\euler_{2\halfpaths}(\lengthprof)$, if Eulerian strategies are best responses to the uniform density, for all $\proba\in (0,1]$ the value of the game is 
\begin{equation*}
\frac{\eulertime(\euler_{2\halfpaths}(\lengthprof))+\proba^{-1}}{2} + \altdiff_\halfpaths.
\end{equation*}
Moreover the \ac{UES} and the \ac{UD} are optimal.

However, Eulerian strategies are not always best responses to the \ac{UD}, as we now argue.

\begin{example}
We study the game played on a parallel Eulerian graph with four parallel paths. Each path $\ppath$ has two edges $\edge_{\ppath1} = \braces{\source,\vertex_{\ppath}}$ and $\edge_{\ppath2} = \braces{\vertex_{\ppath},\sink}$, where $\vertex_{\ppath}$ is the middle vertex of the $i$th path.

Consider the case where Sally visits $\edge_{41}$, $\edge_{42}$ and $\edge_{12}$  before any other edge. 
When she plays an \ac{ES}, this event has positive probability.  
Assume also that, when at $\vertex_1$, the edge $\edge_{12}$ is active but $\edge_{11}$ is not.
At this point she can either wait at $\vertex_2$ until $\edge_{11}$ becomes active in order to follow an \ac{ES}, or she can take $\edge_{12}$, then the first active edge between $\edge_{22}$ and $\edge_{32}$ and continue with $\edge_{21}$ or $\edge_{31}$ respectively. 
Finally, she takes the first active edge between $\edge_{11}$ and the other edge at $\source$ that has not been visited yet, and then visits the two remaining edges as quickly as possible.

Following an \ac{ES} yields the continuation payoff
\begin{align*}
\pay_1 = \frac{1}{5}\left(5 + \frac{11}{\proba} + \frac{4}{1-(1-\proba)^2}\right).
\end{align*}
Following the second strategy yields the continuation payoff
\begin{align*}
\pay_2 = \frac{1}{5}\left(5 + \frac{17}{2\proba} + \frac{8}{1-(1-\proba)^2}\right).
\end{align*}
Hence if $\proba<2/5$, the second strategy yields a lower payoff to Sally than an \ac{ES}.
\end{example}

\subsubsection{The circle}

We now examine the game played on a circle.

\begin{lemma}
\label{le:circle}
If the graph $\graph$ is a circle, then Eulerian strategies are best responses to the uniform density.
\end{lemma}
The proof of \cref{le:circle} is rather straightforward and we omit it. Together with \cref{th:EUequalizingedges,th:UequalizingEU},  \cref{le:circle} yields the following corollary.

\begin{corollary}
The value of the game played on the circle with $\totedges$ edges is 
\[
\val(\proba)=\frac{\eulertime(\graph)+\proba^{-1}}{2} + \altdiff_2 =\frac{1}{1-(1-\proba)^{2}}+\frac{\totedges-1}{2\proba},
\]
for all $\proba\in(0,1]$. Moreover the uniform density and the uniform Eulerian strategy are optimal strategies.
\end{corollary}

\section*{Acknowledgments}
This work has been partly supported by the COST Action CA16228 European Network for Game Theory, by the INdAM-GNAMPA Project 2019 ``Markov chains and games on networks,'' and by the Italian MIUR PRIN 2017 Project ALGADIMAR ``Algorithms, Games, and Digital Markets.''
Marco Scarsini is a member of GNAMPA-INdAM.
Tristan Garrec gratefully acknowledges the hospitality of the Department of Economics and Finance at LUISS, where part of this research was carried out. 

\appendix

\section{Omitted proofs}
\label{se:appendix}

\subsection{Omitted proofs of \cref{se:value}}
The following lemma is a corollary of \citet[Theorem~12]{FlePreSud:AMOfrth}.
\begin{lemma}
\label{le:positivesto}
Positive zero-sum stochastic games with finite state space and action spaces have a value. Moreover the minimizer has an optimal (stationary) strategy.
\end{lemma}

\begin{proof}[Proof of \cref{pr:valueexistence}]
We restate the stochastic search game as a positive zero-sum stochastic game with finite state and action spaces and apply \cref{le:positivesto}. 
The idea is that Sally's stage payoff is $1$ at each stage until she finds Harry, transitioning then to an absorbing state in which the payoff is $0$ forever. The total payoff is  then the sum of the stage payoffs.

Let $\graph = (\vertices,\edges)$ be the underlying graph. 
In order to cast our problem in the framework of \citet{FlePreSud:AMOfrth}, we will use a finite state space $\left(\vertices\times \graphs\times (\edges\cup \{\dagger\})\right)\cup\{\ast\}$, which is larger than the state space $\graphs\times \vertices$, used in \cref{se:model}.
The state $(\vertex_{0},(\vertices,\varnothing),\dagger)$ is the initial state at stage $0$, where $\dagger$ indicated that the hider has not chosen an edge where to hide.
In this state, the finite action space of the searcher is $\neighbor((\vertices,\varnothing),\vertex_{0})=\{\vertex_{0}\}$, the finite action space of the hider is $\edges$ and the payoff is $0$. 
The state $\ast$ is an absorbing state in which the payoff is $0$ forever. In any other state the payoff is $1$.

The state moves from the initial state to $(\vertex_{0},\graph_{1},\edge)$ where $\graph_{1}$ is the graph drawn at stage $1$ and $\edge$ is the edge chosen by the hider (which is fixed for the rest of the game). In any state $(\vertex,\graph',\edge)\in \vertices\times \graphs\times \edges$ the searcher selects $\vertex'\in\neighbor(\graph',\vertex)$ and the hider selects $\edge\in\{\edge\}$. If $\{\vertex,\vertex'\}=\edge$ then the state next moves to the absorbing state $\ast$. If $\{\vertex,\vertex'\}\neq \edge$, the state moves to $(\vertex',\graph'',\edge)$ where $\graph''$ is drawn according to the activation parameter.

Finally, since $\edges$ is finite, the hider has an optimal strategy.
\end{proof}

\subsection{Omitted proofs of \cref{se:tree}}
\begin{proof}[Proof of \cref{le:delta}]

We proceed by induction on the number of edges. The base case is immediate since $\diffedge(\edge_{1})=\diffedge(\edge_{2}) =0$. For the induction step the situation is represented in  \cref{fig:prooflemmadelta}. The vertex $\vertex_{1}$ is the first vertex encountered in $\subtree_{\edge_{1}}$ with two outgoing edges, and similarly for $\vertex_{2}$ and $\subtree_{\edge_{2}}$.

\begin{figure}[ht!]
\begin{center}
\begin{tikzpicture}[scale=1]
\node[above] at (0,0){$\vertex$};
\node[left] at (-2,-2.25){$\vertex_{1}$};
\node[right] at (2,-2.25){$\vertex_{2}$};
\node[left] at (-0.5,-0.5){$\edge_{1}$};
\node[right] at (0.5,-0.5){$\edge_{2}$};
\node[left] at (-1.75,-1.75){$\edge_{1}'$};
\node[right] at (1.75,-1.75){$\edge_{2}'$};
\node[left] at (-2.2,-2.6){$\edge_{11}'$};
\node[right] at (-1.8,-2.6){$\edge_{12}'$};
\node[left] at (1.8,-2.6){$\edge_{21}'$};
\node[right] at (2.2,-2.6){$\edge_{22}'$};

\node at (0,0){$\bullet$};
\node at (-1,-1){$\bullet$};
\node at (1,-1){$\bullet$};
\node at (-1.5,-1.5){$\bullet$};
\node at (1.5,-1.5){$\bullet$};
\node at (-2,-2.25){$\bullet$};
\node at (2,-2.25){$\bullet$};
\node at (-2.4,-3){$\bullet$};
\node at (2.4,-3){$\bullet$};
\node at (-1.6,-3){$\bullet$};
\node at (1.6,-3){$\bullet$};

\draw[-] (0,0) to (-1,-1);
\draw[-] (0,0) to (1,-1);
\draw[dotted] (-1,-1) to (-1.5,-1.5);
\draw[dotted] (1,-1) to (1.5,-1.5);
\draw[-] (-1.5,-1.5) to (-2,-2.25);
\draw[-] (1.5,-1.5) to (2,-2.25);
\draw[-] (-2,-2.25) to (-2.4,-3);
\draw[-] (-2,-2.25) to (-1.6,-3);
\draw[-] (2,-2.25) to (2.4,-3);
\draw[-] (2,-2.25) to (1.6,-3);

\draw [line width=.5pt,dash pattern=on 1pt off 2pt] (-2.4,-3.75) ellipse (0.3cm and 0.75cm);
\draw [line width=.5pt,dash pattern=on 1pt off 2pt] (-1.6,-3.75) ellipse (0.3cm and 0.75cm);
\draw [line width=.5pt,dash pattern=on 1pt off 2pt] (2.4,-3.75) ellipse (0.3cm and 0.75cm);
\draw [line width=.5pt,dash pattern=on 1pt off 2pt] (1.6,-3.75) ellipse (0.3cm and 0.75cm);
\end{tikzpicture}
    \caption{The induction step}
    \label{fig:prooflemmadelta}
\end{center}
\end{figure}
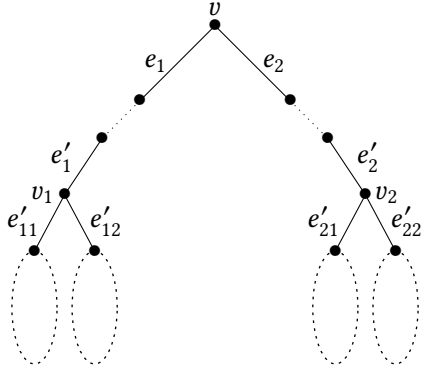
We have
\begin{align*}
\vert\diffedge(\edge_{1})\vert = 
\vert\diffedge(\edge_{1}')\vert &= 
\left\vert\frac{\cycletime(\edge_{11}')}{\cycletime(\edge_{11}')+\cycletime(\edge_{12}')}\diffedge(\edge_{11}') + \frac{\cycletime(\edge_{12}')}{\cycletime(\edge_{11}')+\cycletime(\edge_{12}')}\diffedge(\edge_{12}') + \frac{1}{2}\left(\frac{1}{1-(1-\proba)^{2}}-\frac{1}{\proba}\right)\right\vert\\
&<\frac{1}{2}\max(\cycletime(\edge_{11}'),\cycletime(\edge_{12}'))+\frac{1}{2}\left\vert\frac{1}{1-(1-\proba)^{2}}-\frac{1}{\proba}\right\vert
\end{align*}
by induction, and similarly for $\diffedge(\edge_{2})$. 
Moreover we have 
\begin{align*}
\cycletime(\edge_{1}) > \cycletime(\edge_{1}') = \cycletime(\edge_{11}')+\cycletime(\edge_{12}')+\frac{1}{\proba}+\frac{1}{1-(1-\proba)^{2}},
\end{align*}
and similarly for $\cycletime(\edge_{2})$. Finally,
\begin{align*}
\frac{\vert\diffedge(\edge_{1})\vert+\vert\diffedge(\edge_{2})\vert}{\cycletime(\edge_{1})+\cycletime(\edge_{2})} 
&< \frac{\dfrac{1}{2}\left(\max(\cycletime(\edge_{11}'),\cycletime(\edge_{12}')) + \max(\cycletime(\edge_{21}'),\cycletime(\edge_{22}'))+\dfrac{2}{\proba}-\dfrac{2}{1-(1-\proba^{2})}\right)}{\cycletime(\edge_{11}')+\cycletime(\edge_{12}')+\cycletime(\edge_{21}')+\cycletime(\edge_{22}')+\dfrac{2}{\proba}+\dfrac{2}{1-(1-\proba^{2})}}\\
&<\frac{1}{2}. \qedhere
\end{align*}

\end{proof}

\begin{proof}[Proof of \cref{th:BQDFequalizingleafs}]
We proceed by induction on the number of edges in the tree $\tree$. 
If $\tree$ has only one edge $\edge$, then 
\begin{equation}\label{eq:one-edge}
\pay(\edge,\mixed_{\bdfs}) = \frac{1}{\proba} = \frac{1}{2}\left(\frac{2}{\proba}+ 0\right).
\end{equation} 
Suppose that for any tree that has less edges than $\tree$, the time to reach any leaf edge using the \ac{BDFS} is $\frac{1}{2}\cycletime(\source) + \diffedge(\source)$.

If the origin $\source$ has degree $1$ (as in  \cref{fig:Odegree1}), then, for any leaf edge $\edge$, we have
\begin{equation}\label{eq:pay-bdfs}
\pay(\edge,\mixed_{\bdfs}) = \frac{1}{\proba} + \frac{1}{2}(\cycletime(\vertex))+\diffedge(\vertex) 
= \frac{1}{2}\left(\cycletime(\vertex)+ \frac{2}{\proba}\right)+\diffedge(\vertex)
= \frac{1}{2}\cycletime(\source)+ \diffedge(\source).
\end{equation}

Consider the case where $\source$ has degree $2$ (as in \cref{fig:Odegree2}) and let $\edge_{1}$ be a leaf edge in $\subtree_{\vertex_{1}}$. Then
\begin{align*}
\pay(\edge_{1},\mixed_{\bdfs})
&= (1-\proba)^{2}(1+\pay(\edge_{1},\mixed_{\bdfs})) \\
&\quad+ \proba(1-\proba)\left(1+\frac{1}{2}\cycletime(\vertex_{1})+\diffedge(\vertex_{1}) + 1 + \cycletime(\vertex_{2}) + \frac{2}{\proba} + \frac{1}{2}\cycletime(\vertex_{1})+\diffedge(\vertex_{1})\right)\\
&\quad+ \proba^{2}\left(\bdfs(\edge_{1})\left(1+\frac{1}{2}\cycletime(\vertex_{1})+\diffedge(\vertex_{1})\right) + \bdfs(\edge_{2})\left(1 + \cycletime(\vertex_{2}) + \frac{2}{\proba} + \frac{1}{2}\cycletime(\vertex_{1})+\diffedge(\vertex_{1})\right)\right).
\end{align*}
and
\begin{align*}
\pay(\edge_{1},\mixed_{\bdfs}) (1-(1-\proba)^{2}) 
&= 1 + \proba(1-\proba)\left(\cycletime(\vertex_{1})+\cycletime(\vertex_{2}) + \frac{2}{\proba} + 2\diffedge(\vertex_{1})\right) \\
&\quad + \proba^{2}\left(\frac{1}{2}\cycletime(\vertex_{1})+\diffedge(\vertex_{1}) + \bdfs(\edge_{2})\left(\cycletime(\vertex_{2})+\frac{2}{\proba}\right)\right).
\end{align*}
Furthermore,
\begin{align*}
\cycletime(\source) &= \cycletime(\vertex_{1})+\cycletime(\vertex_{2})+\frac{3}{\proba} + \frac{1}{1-(1-\proba)^{2}} 
\intertext{ and } 
\cycletime(\edge_{1})+\cycletime(\edge_{2}) &= \cycletime(\vertex_{1})+\cycletime(\vertex_{2}) + \frac{4}{\proba} = \cycletime(\source) + \frac{1}{\proba} - \frac{1}{1-(1-\proba)^{2}}.
\end{align*}
Hence, by \cref{le:delta}, for $\proba$ large enough we do not need the projection in \cref{eq:bdsf}, so we have 
\begin{align*}
\pay(\edge_{1},\mixed_{\bdfs}) (1-(1-\proba)^{2}) 
&= 1 + p(1-\proba)\left(\cycletime(\source)-\frac{1}{\proba}-\frac{1}{1-(1-\proba)^{2}} + 2\diffedge(\edge_{1})\right)\\ 
&\quad+ \proba^{2}\left(\frac{1}{2}\left(\cycletime(\edge_{1}) - \frac{2}{\proba}\right) + \diffedge(\edge_{1}) + \left(\frac{1}{2} + \frac{\diffedge(\edge_{2})-\diffedge(\edge_{1})}{\cycletime(\edge_{1})+\cycletime(\edge_{2})} \frac{1-(1-\proba)^{2}}{\proba^{2}}\right)\cycletime(\edge_{2})\right).
\end{align*}
Thus,
\begin{align*}
\pay(\edge_{1},\mixed_{\bdfs}) 
&= \frac{1}{1-(1-\proba)^{2}} + \frac{1}{2}\left(\cycletime(\source) - \frac{1}{\proba}-\frac{1}{1-(1-\proba)^{2}}+2\diffedge(\edge_{1})\right) + \frac{\diffedge(\edge_{2})-\diffedge(\edge_{1})}{\cycletime(\edge_{1})+\cycletime(\edge_{2})}\cycletime(\edge_{2})\\
&= \frac{1}{2}\cycletime(\source) + \frac{1}{2}\left(\frac{1}{1-(1-\proba)^{2}}- \frac{1}{\proba}\right) + \diffedge(\edge_{1})\frac{\cycletime(\edge_{1})}{\cycletime(\edge_{1})+\cycletime(\edge_{2})} + \diffedge(\edge_{2})\frac{\cycletime(\edge_{2})}{\cycletime(\edge_{1})+\cycletime(\edge_{2})}\\
&= \frac{1}{2}\cycletime(\source) + \diffedge(\source). \qedhere
\end{align*}
\end{proof}

\begin{proof}[Proof of \cref{th:EBDequalizingQDF}]
The proof is by induction on the number of edges of the tree $\tree$. If $\tree$  has only one edge, the result is immediate. 
Suppose now that the results holds for any tree with fewer edges than $\tree$.

If the degree of the origin $\source$ is $1$,  the result follows immediately from the induction hypothesis. 
Assume now that the degree of $\source$ is $2$ (as in \cref{fig:Odegree2}). Let $\pure(\vertex_{1})$ and $\pure(\vertex_{2})$ be two \acp{DFS}  on $\subtree_{\vertex_{1}}$ and $\subtree_{\vertex_{2}}$, respectively. 
Let $\pure(\edge_{1})$ be the pure \ac{DFS} on $\tree$ that, when both $\edge_{1}$ and $\edge_{2}$ are active, takes edge $\edge_{1}$ concatenated with $\pure(\vertex_{1})$ and then $\pure(\vertex_{2})$, in case Harry is not found in $\subtree_{\vertex_{1}}$.
The pure strategy $\pure(\edge_{2})$ is defined analogously.
Given a vertex $\vertex$, call $\ebdcond{\vertex}$ the conditional probability measure on $\edges_{\vertex}$ induced by $\ebd$.
Then
\begin{align*}
\pay(\ebd,\pure(\edge_{1})) 
&= (1-\proba)^{2}(1+\pay(\ebd,\pure(\edge_{1}))) \\
&\quad+ \proba\left(\ebd(\edges_{\edge_{1}})(1+\pay(\ebdcond{\vertex_{1}},\pure(\vertex_{1})) + \ebd(\edges_{\edge_{2}})\left(1+\frac{2}{\proba}+\cycletime(\vertex_{1}) + \pay(\ebdcond{\vertex_{2}},\pure(\vertex_{2}))\right)\right) \\
&\quad+ \proba(1-\proba)\left(\ebd(\edges_{\edge_{2}})(1+\pay(\ebdcond{\vertex_{2}},\pure(\vertex_{2})) + \ebd(\edges_{\edge_{1}})\left(1+\frac{2}{\proba}+\cycletime(\vertex_{2}) + \pay(\ebdcond{\vertex_{1}},\pure(\vertex_{1}))\right)\right).
\end{align*}
Hence,
\begin{align*}
\pay(\ebd,\pure(\edge_{1})) = \pay(\ebd,\pure(\edge_{2})) 
&\Longleftrightarrow \ebd(\edges_{\edge_{1}})\left(\frac{2}{\proba}+\cycletime(\vertex_{2})\right) = \ebd(\edges_{\edge_{2}})\left(\frac{2}{\proba}+\cycletime(\vertex_{1})\right)\\
&\Longleftrightarrow \ebd(\edges_{\edge_{1}}) = \frac{\cycletime(\edge_{1})}{\cycletime(\edge_{1})+\cycletime(\edge_{2})}.   \qedhere
\end{align*} 
\end{proof}

\subsection{Omitted proofs of \cref{se:eulerian}}
\begin{proof}[Proof of \cref{th:EUequalizingedges}]
We denote $\edge(\ppath,\edgel)$ the $\edgel$-th edge  of path $\ppath$, starting from the root $\source$. 
We proceed by induction on $\halfpaths$.

Consider that, with probability $(1-\proba)^{2\halfpaths}$ all edges starting from $\source$ are inactive; if this happens, Sally has to wait one turn and her payoff is $(1+\pay(\edge(\ppath,\edgel),\eulerstr))$. 
With probability $1-(1-\proba)^{2\halfpaths}$ at least one edge is active and each of the available edges is chosen with equal probability. 
Given that Harry hides in $\edge(\ppath,\edgel)$, if the chosen path is $\ppath$, then the game ends in $(\edgel-1)/\proba$ units of time.
If the chosen path is $\ppathalt\neq\ppath$, then Sally goes to $\sink$ and the continuation payoff is $\pay_{\ppathalt}(\edge(\ppath,\length_{\ppath}-\edgel+1),\eulerstr)$, where $\pay_{\ppathalt}$ is the payoff of the game played on $\euler_{2\halfpaths-1}(\lengthprof\setminus\length_{\ppathalt})$, in which path $\ppathalt$ has been visited, $\lengthprof\setminus\length_{\ppathalt}$ is the vector $(\length_1,\dots,\length_{\ppathalt-1},\length_{\ppathalt+1},\dots,\length_{2\halfpaths})$ of size $2\halfpaths-1$, and the game starts in $\sink$. 

In formula:
\begin{align*}
\pay(\edge(\ppath,\edgel),\eulerstr) 
&= 
(1-\proba)^{2\halfpaths}(1+\pay(\edge(\ppath,\edgel),\eulerstr)) \\ 
&\quad +\frac{1-(1-\proba)^{2\halfpaths}}{2\halfpaths}\left(1+\frac{\edgel-1}{\proba} + \sum_{\ppathalt\neq \ppath}\left( 1 + \frac{\length_{\ppathalt}-1}{\proba} + \pay_{\ppathalt}(\edge(\ppath,\length_{\ppath}-\edgel+1),\eulerstr) \right)\right).
\end{align*}
The above expression yields
\begin{equation}\label{eq:geijsigma}
\pay(\edge(\ppath,\edgel),\eulerstr) 
= \frac{1}{1-(1-\proba)^{2\halfpaths}} + \frac{1}{2\halfpaths}\left(\frac{\edgel-1}{\proba} + \sum_{\ppathalt\neq \ppath}\left( \frac{\length_{\ppathalt}-1}{\proba} + \pay_{\ppathalt}(\edge(\ppath,\length_{\ppath}-\edgel+1),\eulerstr) \right)\right)
\end{equation}
A similar expression holds for $\pay_{\ppathalt}(\edge(\ppath,\length_{\ppath}-\edgel+1),\eulerstr)$. Plugging it in \cref{eq:geijsigma}, we obtain
\begin{equation}\label{eq:geijsigma-2}
\begin{split}
\pay(\edge(\ppath,\edgel),\eulerstr) =& \frac{1}{1-(1-\proba)^{2\halfpaths}}
+\frac{1}{2\halfpaths}\left(\frac{\edgel-1}{\proba} + \sum_{\ppathalt\neq \ppath}\left( \frac{\length_{\ppathalt}-1}{\proba} + \frac{1}{1-(1-\proba)^{2\halfpaths-1}} \right.\right.\\
&+\left.\left. \frac{1}{2\halfpaths-1}\left(\frac{\length_{\ppath}-\edgel}{\proba} + \sum_{\ppathsec\neq \ppathalt,\ppath}\left(\frac{\length_{\ppathsec}-1}{\proba}+\pay_{\ppathalt,\ppathsec}(\edge(\ppath,\edgel),\eulerstr) \right)\right)\right)\right),
\end{split}
\end{equation}
where $\pay_{\ppathalt,\ppathsec}$ is the payoff of the game played on $\euler_{2(\halfpaths-1)}(\lengthprof\setminus \length_{\ppathalt},\length_{\ppathsec})$, in which both path $\ppathalt$ and path $\ppathsec$ have been visited.
The induction hypothesis is
\begin{equation}\label{eq:g-induction}
\pay_{\ppathalt,\ppathsec}(\edge(\ppath,\edgel),\eulerstr) = \frac{\eulertime(\euler_{2(\halfpaths-1)}(\lengthprof\setminus \length_{\ppathalt},\length_{\ppathsec}))+\proba^{-1}}{2}+\altdiff_{\halfpaths-1}.
\end{equation} 
Therefore, plugging \cref{eq:g-induction} into \cref{eq:geijsigma-2}, we get
\begin{align*}
\pay(\edge(\ppath,\edgel),\eulerstr) 
&= \frac{1}{1-(1-\proba)^{2\halfpaths}} + \frac{2\halfpaths-1}{2\halfpaths}\frac{1}{1-(1-\proba)^{2\halfpaths-1}} \\
&\quad + \frac{1}{2\halfpaths}\left(1+\frac{2\halfpaths-2}{2\halfpaths-1}\right) \sum_{\ppathalt\neq \ppath}\frac{\length_{\ppathalt}-1}{\proba}  + \frac{1}{2\halfpaths}\frac{\length_{\ppath}-1}{\proba} \\
&\quad + \frac{1}{2\halfpaths(2\halfpaths-1)}\sum_{\ppathalt\neq \ppath}\sum_{\ppathsec\neq \ppath,\ppathalt} \left(\frac{\eulertime(\euler_{2(\halfpaths-1)}(\lengthprof\setminus \length_{\ppathalt},\length_{\ppathsec}))+\proba^{-1}}{2}+\altdiff_{\halfpaths-1}\right).
\end{align*}
Furthermore,
\begin{align*}
\sum_{\ppathalt\neq \ppath}\sum_{\ppathsec\neq \ppath,\ppathalt} \left(\eulertime(\euler_{2(\halfpaths-1)}(\lengthprof\setminus \length_{\ppathalt},\length_{\ppathsec}))+\proba^{-1}\right) 
&= \frac{(2\halfpaths-1)(2\halfpaths-2)}{\proba}\\
&\quad + (2\halfpaths-1)(2\halfpaths-2)\sum_{\ppathalt=1}^{2(\halfpaths-1)}\frac{1}{1-(1-\proba)^{\ppathalt}} \\ 
&\quad + (2\halfpaths-1)(2\halfpaths-2)\frac{\length_{\ppath}-1}{\proba}\\
&\quad + (2\halfpaths-2)(2\halfpaths-3)\sum_{\ppathalt\neq \ppath}\frac{\length_{\ppathalt}-1}{\proba}.
\end{align*}
And finally, one obtains the following simplifications
\begin{align*}
\pay(\edge(\ppath,\edgel),\eulerstr) 
&=\frac{1}{1-(1-\proba)^{2\halfpaths}} + \frac{2\halfpaths-1}{2\halfpaths}\frac{1}{1-(1-\proba)^{2\halfpaths-1}} + \frac{2\halfpaths-2}{4\halfpaths}\left(\frac{1}{\proba}+\sum_{\ppathalt=1}^{2(\halfpaths-1)}\frac{1}{1-(1-\proba)^{\ppathalt}}\right)  \\
&\quad + \frac{1}{2\halfpaths}\left(1+\frac{2\halfpaths-2}{2\halfpaths-1}\right) \sum_{\ppathalt\neq \ppath}\frac{\length_{\ppathalt}-1}{\proba}  + \frac{1}{2\halfpaths}\frac{\length_{\ppath}-1}{\proba} + \frac{2\halfpaths-2}{4\halfpaths}\frac{\length_{\ppath}-1}{\proba} \\
&\quad + \frac{(2\halfpaths-2)(2\halfpaths-3)}{4\halfpaths(2\halfpaths-1)}\sum_{k\neq \ppath}\frac{\length_{\ppathalt}-1}{\proba}+ 
 \frac{2\halfpaths-2}{2\halfpaths}\altdiff_{\halfpaths-1}\\
&= \frac{1}{1-(1-\proba)^{2\halfpaths}} + \frac{2\halfpaths-1}{2\halfpaths}\frac{1}{1-(1-\proba)^{2\halfpaths-1}} \\
&\quad + \frac{\halfpaths-1}{2\halfpaths}\left(\frac{1}{\proba}+\sum_{\ppathalt=1}^{2(\halfpaths-1)}\frac{1}{1-(1-\proba)^{\ppathalt}}\right) +  \frac{1}{2}\sum_{\ppathalt=1}^{2\halfpaths} \frac{\length_{\ppathalt}-1}{\proba} +  \frac{\halfpaths-1}{\halfpaths}\altdiff_{\halfpaths-1}\\
&=\frac{\eulertime(\euler_{2\halfpaths}(\lengthprof))+\proba^{-1}}{2} + \altdiff_\halfpaths. \qedhere
\end{align*}
\end{proof}

\begin{proof}[Proof of \cref{th:UequalizingEU}]
The proof is by induction on the number of parallel paths.
Let $\pure$ be a \ac{ES} of Sally, and denote 
\begin{equation*}
\totedges = \sum_{\ppathalt=1}^{2\halfpaths} \length_{\ppathalt}
\end{equation*} 
the number of edges of $\euler_{2\halfpaths}(\lengthprof)$. 
First,
\begin{align*}
\pay(\ud,\pure) = \frac{1}{1-(1-\proba)^{2\halfpaths}} + \frac{1}{2\halfpaths}\left(1-\frac{1}{\totedges}\right)\sum_{\ppathalt=1}^{2\halfpaths} \pay^{\length_\ppathalt-1}(\ud,\pure),
\end{align*}
where $ \pay^{\length_\ppathalt-1}(\ud,\pure)$ is the payoff of the continuation game after one edge of path $\ppathalt$ has been visited.
It is not difficult to prove that 
\begin{align*}
(\totedges - 1)\pay^{\length_\ppathalt-1}(\ud,\pure) = (\totedges - \length_\ppath)\pay_\ppathalt(\ud,\pure) + \frac{\totedges(\length_\ppathalt-1)}{\proba} - \frac{\length_\ppathalt(\length_\ppathalt-1)}{2\proba},
\end{align*}
where $\pay_{\ppathalt}$ is the payoff of the game played on $\euler_{2(\halfpaths-1)}(\lengthprof\setminus \length_{\ppathalt})$, in which path $\ppathalt$ has been visited.
Therefore
\begin{align*}
\pay(\ud,\pure) = \frac{1}{1-(1-\proba)^{2\halfpaths}} + \frac{\totedges}{2\halfpaths\proba} - \frac{1}{\proba} +
\frac{1}{2\halfpaths \totedges}\sum_{\ppathalt=1}^{2\halfpaths} \left(\frac{-\length_{\ppathalt}(\length_{\ppathalt}-1)}{2\proba} + (\totedges - \length_{\ppathalt})\pay_{\ppathalt}(\ud,\pure)\right).
\end{align*}
Computing a similar expression for $\pay_{\ppathalt}(\ud,\pure)$ and plugging it in the above equation one has
\begin{align*}
\pay(\ud,\pure) &= \frac{1}{1-(1-\proba)^{2\halfpaths}} +\frac{2\halfpaths-1}{2\halfpaths}\frac{1}{1-(1-\proba)^{2\halfpaths-1}} + \frac{\totedges}{2\halfpaths\proba} - \frac{1}{\proba} - \frac{2\halfpaths-1}{2\halfpaths\proba}\\
&+\frac{1}{2\halfpaths \totedges}\sum_{\ppathalt=1}^{2\halfpaths}
\left(\frac{-\length_{\ppathalt}(\length_{\ppathalt}-1)}{2\proba} + \frac{(\totedges - \length_{\ppathalt})^2}{(2\halfpaths-1)\proba} 
\right. \\
&+ \left. \frac{1}{(2\halfpaths-1)}\sum_{\ppathsec \neq \ppathalt}\left(\frac{-\length_{\ppathsec}(\length_{\ppathsec}-1)}{2\proba} + (\totedges - \length_{\ppathalt}-\length_{\ppathsec})\pay_{\ppathalt,\ppathsec}(\ud,\pure)\right)\right),
\end{align*}
where $\pay_{\ppathalt,\ppathsec}$ is the payoff of the game played on $\euler_{2(\halfpaths-1)}(\lengthprof\setminus \length_{\ppathalt},\length_{\ppathsec})$, in which both path $\ppathalt$ and path $\ppathsec$ have been visited.
From the induction hypothesis, one has
\begin{align*}
\pay_{\ppathalt,\ppathsec}(\ud,\pure) = \frac{1}{2}\left(\frac{1}{\proba} + \frac{\totedges - \length_{\ppathalt}-\length_{\ppathsec}-2(\halfpaths-1)}{\proba}+ \sum_{l = 1}^{2(\halfpaths-1)}\frac{1}{1-(1-\proba)^l} \right) + \altdiff_{\halfpaths-1}.
\end{align*}
Plugging this expression in the previous equation, one has
\begin{align*}
\pay(\ud,\pure) &= \frac{1}{1-(1-\proba)^{2\halfpaths}} +\frac{2\halfpaths-1}{2\halfpaths}\frac{1}{1-(1-\proba)^{2\halfpaths-1}} + \frac{\totedges}{2\halfpaths\proba} - \frac{1}{\proba} - \frac{2\halfpaths-1}{2\halfpaths\proba} + \frac{\halfpaths-1}{\halfpaths}\altdiff_{\halfpaths-1}\\
&+\frac{\halfpaths-1}{2\halfpaths}\left(\frac{1}{\proba} - \frac{2(\halfpaths-1)}{\proba} + \sum_{l = 1}^{2(\halfpaths-1)}\frac{1}{1-(1-\proba)^l}\right)\\
&+\frac{1}{2\halfpaths \totedges}\sum_{\ppathalt=1}^{2\halfpaths}\left(\frac{-\length_{\ppathalt}(\length_{\ppathalt}-1)}{2\proba} + \frac{(\totedges - \length_{\ppathalt})^2}{(2\halfpaths-1)\proba} 
+ \frac{1}{2\halfpaths-1}\sum_{\ppathsec \neq \ppathalt}\left(\frac{-\length_{\ppathsec}(\length_{\ppathsec}-1)}{2\proba} + \frac{(\totedges - \length_{\ppathalt}-\length_{\ppathsec})^2}{2\proba}\right)\right).
\end{align*}
Furthermore
\begin{multline*}
\frac{1}{2\halfpaths \totedges}\sum_{\ppathalt=1}^{2\halfpaths}\left(\frac{-\length_{\ppathalt}(\length_{\ppathalt}-1)}{2\proba} + \frac{(\totedges - \length_{\ppathalt})^2}{(2\halfpaths-1)\proba} 
+ \frac{1}{2\halfpaths-1}\sum_{\ppathsec \neq \ppathalt}\left(\frac{-\length_{\ppathsec}(\length_{\ppathsec}-1)}{2\proba} + \frac{(\totedges - \length_{\ppathalt}-\length_{\ppathsec})^2}{2\proba}\right)\right)\\
= \frac{1}{2\halfpaths\proba} + \frac{(\halfpaths-1)\totedges}{2\halfpaths\proba}.
\end{multline*}
Finally, we have
\begin{align*}
\pay(\ud,\pure)
&= \frac{1}{1-(1-\proba)^{2\halfpaths}} + \frac{2\halfpaths-1}{2\halfpaths}\frac{1}{1-(1-\proba)^{2\halfpaths-1}} + \frac{\halfpaths-1}{2\halfpaths}\left(\frac{1}{\proba}+\sum_{\ppathalt=1}^{2(\halfpaths-1)}\frac{1}{1-(1-\proba)^{\ppathalt}}\right) \\ 
&\quad +\frac{1}{2}\sum_{\ppathalt=1}^{2\halfpaths} \frac{\length_{\ppathalt}-1}{\proba} +  \frac{\halfpaths-1}{\halfpaths}\altdiff_{\halfpaths-1}\\
&= \frac{\eulertime(\euler_{2\halfpaths}(\lengthprof))+\proba^{-1}}{2} + \altdiff_\halfpaths. \qedhere
\end{align*}
\end{proof}

\bibliographystyle{apalike}
\bibliography{bibsearch}

\begin{thebibliography}{}

\bibitem[Alpern, 2008]{Alp:N2008}
Alpern, S. (2008).
\newblock Hide-and-seek games on a tree to which {E}ulerian networks are
  attached.
\newblock {\em Networks}, 52(3):162--166.

\bibitem[Alpern, 2010]{Alp:SJCO2010}
Alpern, S. (2010).
\newblock Search games on trees with asymmetric travel times.
\newblock {\em SIAM J. Control Optim.}, 48(8):5547--5563.

\bibitem[Alpern, 2011]{Alp:OR2011}
Alpern, S. (2011).
\newblock Find-and-fetch search on a tree.
\newblock {\em Oper. Res.}, 59(5):1258--1268.

\bibitem[Alpern, 2017]{Alp:OR2017}
Alpern, S. (2017).
\newblock Hide-and-seek games on a network, using combinatorial search paths.
\newblock {\em Oper. Res.}, 65(5):1207--1214.

\bibitem[Alpern, 2019]{Alp:TCSfrth}
Alpern, S. (2019).
\newblock Search for an immobile hider in a known subset of a network.
\newblock {\em Theoretical Computer Science}, 794:20 -- 26.
\newblock Special Issue on Theory and Applications of Graph Searching.

\bibitem[Alpern et~al., 2008]{AlpBasGal:IJGT2008}
Alpern, S., Baston, V., and Gal, S. (2008).
\newblock Network search games with immobile hider, without a designated
  searcher starting point.
\newblock {\em Internat. J. Game Theory}, 37(2):281--302.

\bibitem[Alpern et~al., 2009]{AlpBasGal:N2009}
Alpern, S., Baston, V., and Gal, S. (2009).
\newblock Searching symmetric networks with utilitarian-postman paths.
\newblock {\em Networks}, 53(4):392--402.

\bibitem[Alpern and Gal, 2003]{AlpGal:Kluwer2003}
Alpern, S. and Gal, S. (2003).
\newblock {\em The Theory of Search Games and Rendezvous}.
\newblock Kluwer Academic Publishers, Boston, MA.

\bibitem[Alpern and Lidbetter, 2013]{AlpLid:OR2013}
Alpern, S. and Lidbetter, T. (2013).
\newblock Mining coal or finding terrorists: the expanding search paradigm.
\newblock {\em Oper. Res.}, 61(2):265--279.

\bibitem[Alpern and Lidbetter, 2014]{AlpLid:MOR2014}
Alpern, S. and Lidbetter, T. (2014).
\newblock Searching a variable speed network.
\newblock {\em Math. Oper. Res.}, 39(3):697--711.

\bibitem[Alpern and Lidbetter, 2015]{AlpLid:OR2015}
Alpern, S. and Lidbetter, T. (2015).
\newblock Optimal trade-off between speed and acuity when searching for a small
  object.
\newblock {\em Oper. Res.}, 63(1):122--133.

\bibitem[Alpern and Lidbetter, 2019a]{AlpLid:AOR2019}
Alpern, S. and Lidbetter, T. (2019a).
\newblock Approximate solutions for expanding search games on general networks.
\newblock {\em Ann. Oper. Res.}, 275(2):259--279.

\bibitem[Alpern and Lidbetter, 2019b]{AlpLid:unpublished}
Alpern, S. and Lidbetter, T. (2019b).
\newblock Search and delivery man problems: When are depth-first paths optimal?
\newblock {\em arXiv:1910.13178}.

\bibitem[Anderson and Aramendia, 1990]{AndAra:N1990}
Anderson, E.~J. and Aramendia, M.~A. (1990).
\newblock The search game on a network with immobile hider.
\newblock {\em Networks}, 20(7):817--844.

\bibitem[Aumann, 1964]{Aum:AGT1964}
Aumann, R.~J. (1964).
\newblock Mixed and behavior strategies in infinite extensive games.
\newblock In {\em Advances in {G}ame {T}heory}, pages 627--650. Princeton Univ.
  Press, Princeton, N.J.

\bibitem[Baston et~al., 1990]{BasBosRuc:JOTA1990}
Baston, V.~J., Bostock, F.~A., and Ruckle, W.~H. (1990).
\newblock The gold-mine game.
\newblock {\em J. Optim. Theory Appl.}, 64(3):641--650.

\bibitem[Beck and Newman, 1970]{BecNew:IJM1970}
Beck, A. and Newman, D.~J. (1970).
\newblock Yet more on the linear search problem.
\newblock {\em Israel J. Math.}, 8:419--429.

\bibitem[Berry and Mensch, 1986]{BerMen:OR1986}
Berry, D.~A. and Mensch, R.~F. (1986).
\newblock Discrete search with directional information.
\newblock {\em Oper. Res.}, 34(3):470--477.

\bibitem[Boczkowski et~al., 2018]{BocKorRod:ESA2018}
Boczkowski, L., Korman, A., and Rodeh, Y. (2018).
\newblock Searching a tree with permanently noisy advice.
\newblock In {\em 26th {E}uropean {S}ymposium on {A}lgorithms}, volume 112 of
  {\em LIPIcs. Leibniz Int. Proc. Inform.}, pages Art. No. 54, 13. Schloss
  Dagstuhl. Leibniz-Zent. Inform., Wadern.

\bibitem[Bollob\'{a}s, 2001]{Bol:CUP2001}
Bollob\'{a}s, B. (2001).
\newblock {\em Random Graphs}.
\newblock Cambridge University Press, Cambridge, second edition.

\bibitem[Bollob{\'a}s et~al., 2013]{BolKunLea:JCT2013}
Bollob{\'a}s, B., Kun, G., and Leader, I. (2013).
\newblock Cops and robbers in a random graph.
\newblock {\em Journal of Combinatorial Theory, Series B}, 103:226--236.

\bibitem[Bollob\'{a}s and Riordan, 2006]{Bol:CUP2006}
Bollob\'{a}s, B. and Riordan, O. (2006).
\newblock {\em Percolation}.
\newblock Cambridge University Press, New York.

\bibitem[Bostock, 1984]{Bos:SIAMJADM1984}
Bostock, F.~A. (1984).
\newblock On a discrete search problem on three arcs.
\newblock {\em SIAM J. Algebraic Discrete Methods}, 5(1):94--100.

\bibitem[Cao, 1995]{Cao:EJOR1995}
Cao, B. (1995).
\newblock Search-hide games on trees.
\newblock {\em European J. Oper. Res.}, 80(1):175--183.

\bibitem[Dagan and Gal, 2008]{DagGal:N2008}
Dagan, A. and Gal, S. (2008).
\newblock Network search games, with arbitrary searcher starting point.
\newblock {\em Networks}, 52(3):156--161.

\bibitem[Day and Falgas-Ravry, 2018]{DayFal:arXiv2018}
Day, A.~N. and Falgas-Ravry, V. (2018).
\newblock Maker-breaker percolation games {I}: crossing grids.
\newblock arXiv:1810.05190.

\bibitem[Efron, 1964]{Efr:JSIAM1964}
Efron, B. (1964).
\newblock Optimum evasion versus systematic search.
\newblock {\em J. Soc. Indust. Appl. Math.}, 12:450--457.

\bibitem[Erd\H{o}s and R\'{e}nyi, 1959]{ErdRen:PMD1959}
Erd\H{o}s, P. and R\'{e}nyi, A. (1959).
\newblock On random graphs. {I}.
\newblock {\em Publ. Math. Debrecen}, 6:290--297.

\bibitem[Erd\H{o}s and R\'{e}nyi, 1960]{ErdRen:MTAMKIK1960}
Erd\H{o}s, P. and R\'{e}nyi, A. (1960).
\newblock On the evolution of random graphs.
\newblock {\em Magyar Tud. Akad. Mat. Kutat\'{o} Int. K\"{o}zl.}, 5:17--61.

\bibitem[Erd\H{o}s and R\'{e}nyi, 1961]{ErdRen:BIIS1961}
Erd\H{o}s, P. and R\'{e}nyi, A. (1961).
\newblock On the evolution of random graphs.
\newblock {\em Bull. Inst. Internat. Statist.}, 38:343--347.

\bibitem[Flesch et~al., 2018]{FlePreSud:AMOfrth}
Flesch, J., Predtetchinski, A., and Sudderth, W. (2018).
\newblock Positive zero-sum stochastic games with countable state and action
  spaces.
\newblock {\em Applied Mathematics \& Optimization}, pages 1--18.

\bibitem[Gal, 1972]{Gal:IJM1972}
Gal, S. (1972).
\newblock A general search game.
\newblock {\em Israel J. Math.}, 12:32--45.

\bibitem[Gal, 1974]{Gal:SJAM1974}
Gal, S. (1974).
\newblock A discrete search game.
\newblock {\em SIAM J. Appl. Math.}, 27:641--648.

\bibitem[Gal, 1979]{Gal:SJCO1979}
Gal, S. (1979).
\newblock Search games with mobile and immobile hider.
\newblock {\em SIAM J. Control Optim.}, 17(1):99--122.

\bibitem[Gal, 1980]{Gal:AP1980}
Gal, S. (1980).
\newblock {\em Search games}.
\newblock Academic Press, Inc., New York-London.

\bibitem[Gal, 2000]{Gal:IJGT2000}
Gal, S. (2000).
\newblock On the optimality of a simple strategy for searching graphs.
\newblock {\em Internat. J. Game Theory}, 29(4):533--542 (2001).

\bibitem[Gal and Chazan, 1976]{GalCha:SJAM1976}
Gal, S. and Chazan, D. (1976).
\newblock On the optimality of the exponential functions for some minimax
  problems.
\newblock {\em SIAM J. Appl. Math.}, 30(2):324--348.

\bibitem[Gittins and Roberts, 1979]{GitRob:NRLQ1979}
Gittins, J.~C. and Roberts, D.~M. (1979).
\newblock The search for an intelligent evader concealed in one of an arbitrary
  number of regions.
\newblock {\em Naval Res. Logist. Quart.}, 26(4):651--666.

\bibitem[Glazebrook et~al., 2019]{GlaClaLin:ORfrth}
Glazebrook, K., Clarkson, J., and Lin, K. (2019).
\newblock Fast or slow: search in discrete locations with two search modes.
\newblock {\em Oper. Res.}, forthcoming.

\bibitem[Grimmett, 1999]{Gri:Springer1999}
Grimmett, G. (1999).
\newblock {\em Percolation}.
\newblock Springer-Verlag, Berlin, second edition.

\bibitem[Hohzaki, 2016]{Hoh:JORSJ2016}
Hohzaki, R. (2016).
\newblock Search games: literature and survey.
\newblock {\em J. Oper. Res. Soc. Japan}, 59(1):1--34.

\bibitem[Holroyd et~al., 2019]{HolMarMar:PTRFfrth}
Holroyd, A.~E., Marcovici, I., and Martin, J.~B. (2019).
\newblock Percolation games, probabilistic cellular automata, and the hard-core
  model.
\newblock {\em Probability Theory and Related Fields}, 174(3-4):1187--1217.

\bibitem[Isaacs, 1965]{Isa:Wiley1996}
Isaacs, R. (1965).
\newblock {\em Differential Games. {A} Mathematical Theory with Applications to
  Warfare and Pursuit, Control and Optimization}.
\newblock John Wiley \& Sons, Inc., New York-London-Sydney.

\bibitem[Jotshi and Batta, 2008]{JosBat:EJOR2008}
Jotshi, A. and Batta, R. (2008).
\newblock Search for an immobile entity on a network.
\newblock {\em European J. Oper. Res.}, 191(2):347--359.

\bibitem[Kikuta, 1990]{Kik:JORSJ1990}
Kikuta, K. (1990).
\newblock A hide and seek game with traveling cost.
\newblock {\em J. Oper. Res. Soc. Japan}, 33(2):168--187.

\bibitem[Kikuta, 1991]{Kik:JORSJ1991}
Kikuta, K. (1991).
\newblock A search game with traveling cost.
\newblock {\em J. Oper. Res. Soc. Japan}, 34(4):365--382.

\bibitem[Kikuta, 2004]{Kik:NRL2004}
Kikuta, K. (2004).
\newblock A search game on a cyclic graph.
\newblock {\em Naval Res. Logist.}, 51(7):977--993.

\bibitem[Neuts, 1963]{Neu:JSIAM1963}
Neuts, M.~F. (1963).
\newblock A multistage search game.
\newblock {\em J. Soc. Indust. Appl. Math.}, 11:502--507.

\bibitem[Pavlovi\'{c}, 1995]{Pav:NRL1995}
Pavlovi\'{c}, L. (1995).
\newblock A search game on the union of graphs with immobile hider.
\newblock {\em Naval Res. Logist.}, 42(8):1177--1189.

\bibitem[Reijnierse and Potters, 1993]{ReiPot:IJGT1993}
Reijnierse, J.~H. and Potters, J. A.~M. (1993).
\newblock Search games with immobile hider.
\newblock {\em Internat. J. Game Theory}, 21(4):385--394.

\bibitem[Roberts and Gittins, 1978]{RobGit:NRLQ1978}
Roberts, D.~M. and Gittins, J.~C. (1978).
\newblock The search for an intelligent evader; strategies for searcher and
  evader in the two-region problem.
\newblock {\em Naval Res. Logist. Quart.}, 25(1):95--106.

\bibitem[Sakaguchi, 1973]{Sak:JORSJ1973}
Sakaguchi, M. (1973).
\newblock Two-sided search games.
\newblock {\em J. Operations Res. Soc. Japan}, 16:207--225.

\bibitem[Sorin, 2002]{Sor:Springer2002}
Sorin, S. (2002).
\newblock {\em A First Course on Zero-Sum Repeated Games}.
\newblock Springer-Verlag, Berlin.

\bibitem[Subelman, 1981]{Sub:JAP1981}
Subelman, E.~J. (1981).
\newblock A hide-search game.
\newblock {\em J. Appl. Probab.}, 18(3):628--640.

\bibitem[van~der Hofstad, 2017]{Hof:CUP2017}
van~der Hofstad, R. (2017).
\newblock {\em Random Graphs and Complex Networks. {V}ol. 1}.
\newblock Cambridge University Press, Cambridge.

\bibitem[von Neumann, 1953]{Neu:CTG1953}
von Neumann, J. (1953).
\newblock A certain zero-sum two-person game equivalent to the optimal
  assignment problem.
\newblock In {\em Contributions to the Theory of Games, vol. 2}, Annals of
  Mathematics Studies, no. 28, pages 5--12. Princeton University Press,
  Princeton, N. J.

\bibitem[von Stengel and Werchner, 1997]{SteWer:DAM1997}
von Stengel, B. and Werchner, R. (1997).
\newblock Complexity of searching an immobile hider in a graph.
\newblock {\em Discrete Appl. Math.}, 78(1-3):235--249.

\end{thebibliography}

\end{document}